\documentclass[11pt]{article}
\usepackage[margin=1in]{geometry}

\usepackage{hyperref}
\usepackage{amsmath}
\usepackage{amssymb}
\usepackage{amsfonts}
\usepackage{amsthm}
\usepackage{thm-restate}
\usepackage{subcaption}
\usepackage{xcolor}
\usepackage{graphicx}
\usepackage{cleveref}
\usepackage{MnSymbol}

\newtheorem{obs}{Observation}[section]
\crefname{obs}{Observation}{Observations}
\newtheorem{remark}{Remark}[section]

\newtheorem{basicprop}{Basic Property}[section]
\newtheorem{corollary}{Corollary}[section]
\newtheorem{lemma}{Lemma}[section]

\newcommand{\bestparameterizedsolution}{O(k(|V| + |E|)\log{|V|})}
\newcommand{\tmp}{\mathcal{T}}
\newcommand{\pathcover}{\mathcal{P}}
\newcommand{\flowG}{\mathcal{G}}
\newcommand{\flowV}{\mathcal{V}}
\newcommand{\flowE}{\mathcal{E}}

\newcommand{\Nzero}{\mathbb{N}_{0}}

\newcommand{\invA}{\textbf{Invariant A}}
\newcommand{\invB}{\textbf{Invariant B}}
\newcommand{\invC}{\textbf{Invariant C}}
\newcommand{\invAB}{\textbf{Invariants A and B}}

\newcommand{\residual}[2]{\mathcal{R}(#1, #2)}
\newcommand{\width}{\mathsf{width}}
\newcommand{\blue}{$blue$}
\newcommand{\red}{$red$}
\newcommand{\purple}{$purple$}
\newcommand{\processed}{processed}

\title{Minimum Path Cover in Parameterized Linear Time\thanks{A preliminary version of this work was published in the proceedings of SODA 2022~\cite{caceres2022sparsifying}. This work was partially funded by the US Fulbright program, the Fulbright Finland Foundation, the Helsinki Institute for Information Technology (HIIT), the US National Science Foundation (award DBI-1759522), the European Research Council (ERC) under the European Union's Horizon 2020 research and innovation programme (grant agreement No.~851093, SAFEBIO), and the Academy of Finland (grants No.~322595, 328877).}}

\author{Manuel Cáceres\thanks{Department of Computer Science, University of Helsinki, Finland, \texttt{manuel.caceresreyes@helsinki.fi}.}
\and Massimo Cairo\thanks{Department of Computer Science, University of Helsinki, Finland, \texttt{cairomassimo@gmail.com}.}
\and Brendan Mumey\thanks{Gianforte School of Computer Science, Montana State University, USA, \texttt{brendan.mumey@montana.edu}.}
\and Romeo Rizzi\thanks{Department of Computer Science, University of Verona, Italy, \texttt{romeo.rizzi@univr.it}.}
\and Alexandru~I.~Tomescu\thanks{Department of Computer Science, University of Helsinki, Finland, \texttt{alexandru.tomescu@helsinki.fi}.}}

\begin{document}

\date{}

\maketitle

% REQUIRED
\begin{abstract} A \emph{minimum path cover} (MPC) of a directed acyclic graph (DAG) $G = (V,E)$ is a minimum-size set of paths that together cover all the vertices of the DAG. Computing an MPC is a basic polynomial problem, dating back to Dilworth's and Fulkerson's results in the 1950s. Since the size $k$ of an MPC (also known as the \emph{width}) can be small in practical applications, research has also studied algorithms whose running time is parameterized on $k$.

 We obtain a new MPC parameterized algorithm for DAGs running in time $O(k^2|V| + |E|)$. Our algorithm is the first solving the problem in parameterized linear time. Additionally, we obtain an edge sparsification algorithm preserving the width of a DAG but reducing $|E|$ to less than $2|V|$.
 This algorithm runs in time $O(k^2|V|)$ and requires an MPC of a DAG as input, thus its total running time is the same as the running time of our MPC algorithm.
 \end{abstract}

% REQUIRED
% \begin{keywords}
% Minimum Path Cover, Maximum Antichain, Directed Acyclic Graph, Poset, Maximum Flow, Edge Sparsification, Parameterized Algorithms
% \end{keywords}

% REQUIRED
% \begin{MSCcodes}
% 06A06, 68Q25, 68W99
% \end{MSCcodes}

\section{Introduction}

A \emph{Minimum Path Cover (MPC)} of a (directed) graph $G = (V, E)$ is a minimum-sized set of paths such that every vertex appears in at least one path of the set. While computing an MPC is NP-hard in general, it is a classic result, dating back to Dilworth~\cite{dilworth2009decomposition} and Fulkerson~\cite{fulkerson1956note}, that this can be done in polynomial time on directed acyclic graphs (\emph{DAGs}). Computing an MPC of a DAG has applications in various fields. In bioinformatics, it allows efficient solutions to the problems of multi-assembly~\cite{eriksson2008viral,trapnell2010transcript,rizzi2014complexity,chang2015bridger,liu2017strawberry,caceres2021safety,caceres2022width}, perfect phylogeny haplotyping~\cite{bonizzoni2007linear,gramm2007haplotyping}, and alignment to pan-genomes~\cite{makinen2019sparse,Ma2022.01.07.475257,cotumaccio2021indexing,cotumaccio2022graphs,cotumaccio2021regular,d2021co}. Other examples include scheduling~\cite{colbourn1985minimizing,desrosiers1995time,bunte2009overview,van2016precedence,zhan2016graph,marchal2018parallel}, computational logic~\cite{bova2015model,gajarsky2015fo}, distributed computing~\cite{tomlinson1997monitoring,ikiz2006efficient}, databases~\cite{Jagadish90}, evolutionary computation~\cite{jaskowski2011formal}, program testing~\cite{ntafos1979path}, cryptography~\cite{mackinnon1985optimal}, and programming languages~\cite{kowaluk2008path}. Since in many of these applications the size $k$ (number of paths, also known as \emph{width}) of an MPC is bounded, research has 
% not only focused in classical solutions but 
also focused in solutions whose running time is parameterized by $k$\footnote{Algorithms of this class do not know the value of $k$ in advance.}. This approach is also related to the line of research ``FPT inside P''~\cite{giannopoulou2017polynomial,caceres2021safety,fomin2018fully,koana2021data,abboud2016approximation,caceres2021a,caceres2022sparsifying,makinen2019sparse,Ma2022.01.07.475257} of finding natural parameterizations for problems already in P. 

MPC algorithms can be divided into those based on a reduction to maximum matching~\cite{fulkerson1956note}, and those based on a reduction to minimum flow~\cite{ntafos1979path}. The former compute an MPC of a \emph{transitive} DAG by finding a maximum matching in a bipartite graph with $2|V|$ vertices and $|E|$ edges. Thus, one can compute an MPC of a transitive DAG in time $O(\sqrt{|V|} |E|)$ with the Hopcroft-Karp algorithm \cite{hopcroft1973n}. Further developments of this idea include the $O(k|V|^2)$-time algorithm of Felsner et~al.~\cite{felsner2003recognition}, and the $O(|V|^2 + k\sqrt{k}|V|)$ and $O(\sqrt{|V|}|E| + k\sqrt{k}|V|)$-time algorithms of Chen and Chen~\cite{chen2008efficient,chen2014graph}.

The reduction to minimum flow consists in building a flow network $\flowG$ from $G$, where a global source $s$ and global sink $t$ are added, and each vertex $v$ of $G$ is split into an edge $(v^{in}, v^{out})$ of $\flowG$ with a demand (lower bound) of one unit of flow (see \Cref{sec:preliminaries} for details). A minimum-valued (integral) flow of $\flowG$ corresponds to an MPC of $G$, which can be obtained by decomposing the flow into paths. This reduction (or similar) has been used several times in the literature to compute an MPC (or similar object)~\cite{ntafos1979path,mohring1985algorithmic,gavril1987algorithms,Jagadish90,ciurea2004sequential,rademaker2012optimal,pijls2013another,marchal2018parallel}, and it is used in the recent $\bestparameterizedsolution$-time solution of M\"akinen et~al.~\cite{makinen2019sparse}. Furthermore, by noting that a vertex cannot belong to more than $|V|$ paths in an MPC, the problem can be reduced to maximum flow with capacities at most $|V|$ (see for example \cite[Theorem 3.9.1]{bang2008digraphs}). As an example, using the Goldberg-Rao algorithm~\cite{goldberg1998beyond} the problem can be solved in time $\widetilde{O}(|E|\min(|E|^{1/2}, |V|^{2/3})+||\pathcover||)$ ($||\pathcover||$ is the total length of the output cover, this term is needed for decomposing the flow into an MPC, see e.g.~\cite[Lemma 1.11 (full version)]{kogan2022beating}). More recent maximum flow algorithms~\cite{lee2014path,madry2016computing,liu2020faster,kathuria2020unit,van2021minimum,gao2021fully,chen2022maximum} provide an abundant options of trade-offs, though none of them leads to a parameterized linear-time solution for the MPC problem. In particular, by applying the recent breakthrough result of Chen et.~al~\cite{chen2022maximum} the problem is solved in almost optimal $\widetilde{O}(|E|^{1+o(1)} + ||\pathcover||)$ time w.h.p. 

A problem closely related to MPC is to find a \emph{minimum chain cover} (MCC), which is a vertex-disjoint minimum-size set of paths in the \emph{transitive closure} of the DAG. An MCC can be obtained from an MPC $\pathcover$ in time $O(||\pathcover||)$ by simply removing the repeated vertices in the paths (see e.g.~\cite[Lemma 1.12 (full version)]{kogan2022beating}). As such, the running times described before also apply for MCC. Recently, Kogan and Parter~\cite{kogan2022beating} presented an algorithm for MCC running in time $\widetilde{O}(|E|+|V|^{3/2})$, avoiding the $||\pathcover||$ term of MPC-based algorithms by using \emph{reachability shortcuts}.

%\ariel{this is the picture using new max flow algorithms $O(|E|^{4/3+o(1)}|V|^{1/3})$~\cite{kathuria2020unit} $\widetilde{O}(|E||V|^{1/2})$~\cite{lee2014path},$\widetilde{O}(|E|^{10/7}|V|^{1/7})$~\cite{madry2016computing}, randomized $\widetilde{O}(|E|+|V|^{3/2})$~\cite{van2021minimum},$\widetilde{O}(|E|^{3/2-1/328})$~\cite{gao2021fully}, randomized $O(|E|^{11/8+o(1)}|V|^{1/4})$~\cite{liu2020faster}}

\subsection{Techniques and Results}
In this paper we propose the first parameterized linear time algorithm to compute an MPC of a DAG. More formally, we obtain the following theorem.
\begin{restatable}{thm}{improvedmpcflow}
\label{thm:improved-mpc-flow}
Given a width-$k$ DAG $G = (V,E)$, we compute an MPC in time ${O(k^2|V|+|E|)}$.
\end{restatable}

To obtain this result we further develop and interleave the use of two previously known techniques, namely, \emph{transitive sparsification} and \emph{shrinking}.

\emph{Transitive sparsification} consists in the removal of some transitive edges\footnote{An edge $(u, v)$ is transitive if after removing it $u$ still reaches $v$.}, which preserves the reachability among vertices, and thus the width of the DAG\footnote{Every edge in an MPC removed by a transitive sparsification can be \emph{re-routed} through an alternative path.}. We sparsify the edges to $O(k|V|)$ only, in overall $O(k|V| + |E|)$ time, obtaining a linear dependency on $|E|$ in our running times. Our idea is inspired by the work of Jagadish~\cite{Jagadish90}, which proposed a compressed index for answering reachability queries in constant time: for each vertex $v$ and path $P$ of an MPC, it stores the last vertex in $P$ that reaches $v$, thus storing $k$ pointers per vertex.
%the last vertex $u$ in $P$ that reaches $v$, denoted as $u = \texttt{lastToReach}(P, v)$. The \texttt{lastToReach} table has size $O(k|V|)$ and it can be computed in $O(k|E|)$ time by dynamic programming~\cite{makinen2019sparse}. As such, since our MPC algorithm reduces $|E|$ to $O(k|V|)$ we can compute the index in the same running time.
% \begin{restatable}{lemma}{reachabilityindexdag}
% \label{lemma:reachability-index-dag}
% Given a DAG $G = (V,E)$ of width $k$, we build an index of size $O(k|V|)$ in time ${O(k^2|V|+|E|)}$, answering reachability queries between pairs of vertices in constant time.
% \end{restatable}
% Moreover, the problem of reachability queries in general directed graphs can be self-reduced to the \emph{condensation}\footnote{The condensation is also known as the DAG of the strongly connected components of the graph.} of the graph (see e.g.~\cite{Ma2022.01.07.475257}), allowing us to solve the problem in the general case.
% \begin{restatable}{thm}{reachabilityindex}
% \label{thm:reachability-index}
% Given a directed graph $G = (V,E)$ of width $k$, we build an index of size $O(k|V|)$ in time ${O(k^2|V|+|E|)}$, answering reachability queries between pairs of vertices in constant time.
% \end{restatable}
However, three issues arise when trying to apply this idea \emph{inside} an MPC algorithm: (i) it is dependent on an initial MPC (whereas we are trying to compute one), (ii) its construction algorithm~\cite{makinen2019sparse} introduces a $k$ multiplicative factor to $|E|$ in the running time, and (iii) transitive edges represented in the index might not be present in the DAG. We address (i) by using a suboptimal (but yet bounded) path cover whose gradual computation is interleaved with transitive sparsifications, and we address (ii) and (iii) by keeping only $O(k)$ incoming edges per vertex in a \emph{single linear pass} over the edges. 

By \emph{shrinking} we refer to the process of transforming an arbitrary path cover into an MPC. As discussed before, \emph{shrinking} can be implemented as a reduction to minimum flow by creating a flow network $\flowG$ from the DAG and the path cover to be shrunk. While previous approaches have used shrinking as a separate black-box, we instead modify it by guiding the search for decrementing paths bounding the total search time to parameterized linear.

C\'aceres et.~al~\cite{caceres2022sparsifying}\footnote{This algorithm was part of the preliminary version of this paper, but we decided to remove it since it is beaten by our improved parameterized linear time algorithm.} proposed a divide and conquer MPC algorithm running in $O(k^2|V|\log{|V|}+|E|)$ time, which interleaves applications of sparsification and shrinking. More specifically, the algorithm divides the DAG in two halves according to a topological order of the vertices, recursively solves each half, sparsifies the edges between both halves, and shrinks the union of both MPCs.

Our MPC algorithm works on top of the minimum flow reduction, but instead of running a minimum flow algorithm and then extracting the corresponding paths (as previous approaches do~\cite{ntafos1979path,mohring1985algorithmic,Jagadish90,ciurea2004sequential,rademaker2012optimal,pijls2013another,marchal2018parallel,kogan2022beating}), it processes the vertices in topological order, and incrementally maintains a minimum flow representing an MPC $\pathcover$ of the corresponding induced subgraph. When a new vertex $v$ is processed, we sparsify the edges incoming to $v$ to at most $k$. After that, we attempt to shrink $\pathcover \cup \{(v)\}$ by searching for a single decrementing path in the corresponding residual graph. The search is guided by assigning an integer level to each vertex and traversing the graph in a \emph{layered} manner (see \Cref{sec:layered-traversal}). The parameterized linear time is achieved by the combination of step-by-step transitive sparsification and level reassignment, which allow us to bound the running time to $O(f(k))$ per vertex. In a first version of our algorithm\footnote{A preliminary version of this first version was presented by the authors in SODA22~\cite{caceres2022sparsifying}.} we explicitly maintain a corresponding MPC $\pathcover$ during the algorithm to effectively sparsify, obtaining the following result.

\begin{restatable}{thm}{mpcflow}
\label{thm:mpc-flow}
Given a width-$k$ DAG $G = (V,E)$, we compute an MPC in time ${O(k^3|V|+|E|)}$.
\end{restatable}

To obtain $O(k^2|V|+|E|)$ time we remove the maintenance of the path cover $\pathcover$ from the algorithm while still being able to sparsify the edges at each step. We note that (for sparsifying) it suffices that every vertex ``knows'' one path that contains it, thus decreasing the information stored locally in a vertex from $O(k)$ to $O(1)$, which globally decreases the term $k^3$ to $k^2$ in the running time of the algorithm. Analogously, this improvement can be seen as replacing the maintenance of an MPC by that of an MCC, which suffices to sparsify transitive edges. The key idea to achieve this is to efficiently maintain \emph{back links} from every vertex $v$ to the previous vertex in some path of $\pathcover$ that has a lower level than that of $v$. In fact, our level assignment implicitly maintains a structure of \emph{antichains} (see \Cref{sec:antichain-structure}) that \emph{sweep} the graph during the algorithm, and the back links point to vertices in these antichains.

An \emph{antichain} is a set of pairwise non-reachable vertices, and it is a well-known result, due to Dilworth~\cite{dilworth2009decomposition}, that the maximum size of an antichain equals the size of an MPC. The high-level idea of maintaining a collection of antichains has been used previously by Felsner et~al.~\cite{felsner2003recognition} and C\'aceres et~al.~\cite{caceres2021a} for the problem of computing a maximum antichain. However, apart from being restricted to this related problem, these two approaches have intrinsic limitations. More precisely, Felsner et~al.~\cite{felsner2003recognition} maintain a \emph{tower of right-most antichains} for \emph{transitive} DAGs and $k \leq 4$, mentioning that ``the case $k=5$ already seems to require an unpleasantly involved case analysis''~\cite[p.~359]{felsner2003recognition}. C\'aceres et~al.~\cite{caceres2021a} overcome this by maintaining $O(2^k)$ many \emph{frontier antichains}, and obtaining a parameterized $O(k^24^k|V| + k2^k|E|)$ time maximum antichain algorithm.

Based on the relation between maximum one-way cuts in the minimum flow reduction and maximum antichains in the original DAG (see for example~\cite{mohring1985algorithmic,pijls2013another,marchal2018parallel}), we obtain algorithms computing a maximum antichain from any of the existing algorithms, preserving their running times (see \Cref{thm:main-antichain-minflow,thm:main-antichain}). In particular, by using our MPC algorithm we obtain an exponential improvement on the function of $k$ of the algorithm of C\'aceres et~al.~\cite{caceres2021a}.

Our last result in \Cref{sec:edge-thinning} is a structural result concerning the problem of \emph{edge sparsification} preserving the width of the DAG. Edge sparsification is a general concept that consists in finding spanning subgraphs (usually with significantly less edges) while (approximately) preserving certain property of the graph. For example, \emph{spanners} are distance preserving (up to multiplicative factors) sparsifiers, and it is a well-known result that $(1+\epsilon)$ cut sparsifiers can be computed efficiently~\cite{benczur1996approximating}. We show that if the property we want to maintain is the (exact) width of a DAG, then its edges can be sparsified to less than~$2|V|$. Moreover, we show that such sparsification is asymptotically tight (\Cref{remark:2-tight}), and it can be computed in $O(k^2|V|)$ time if an MPC is given as additional input. Therefore, by our MPC algorithm we obtain the following result.

\begin{corollary}
\label{cor:main-sparsification}
Given a width-$k$ DAG $G = (V,E)$, we compute a spanning subgraph $G' = (V, E')$ of $G$ with $|E'| < 2|V|$ and width $k$ in time ${O(k^2|V|+|E|)}$.
\end{corollary}

The main ingredient to obtain this result is an algorithm for transforming any path cover into one of the same size using less than~$2|V|$ distinct edges, a surprising structural result.

\begin{restatable}{thm}{edgethinning}
\label{thm:edge-thinning}
Let $G = (V,E)$ be a DAG, and let $\pathcover, |\pathcover|=t$ be a path cover of $G$. We compute, in $O(t^2|V|)$ time, a path cover $\pathcover',|\pathcover'|=t$, whose number of \emph{distinct} edges is less than $2|V|$.
\end{restatable}

We obtain \Cref{cor:main-sparsification} by using \Cref{thm:edge-thinning} 
with an MPC and defining $E'$ as the edges in $\pathcover'$. Our approach adapts the techniques used by Schrijver~\cite{schrijver1998bipartite} for finding a perfect matching in a regular bipartite graph. In our algorithm, we repeatedly search for undirected cycles~$C$ of edges joining vertices of high degree (in the graph induced by the path cover), and \emph{splice} paths along $C$ (according to the \emph{multiplicty} of the edges of $C$) to remove edges from the path cover. By \emph{splicing} we refer to the general process of reconnecting paths in a path cover so that (after splicing) at least one of them contains a certain path $D$ as a subpath, while working in time proportional to $|D|$ (see \Cref{sec:splicing}).

% Maybe put the paper structure back in when is clear what this structure is
% \paragraph{Paper structure.} \Cref{sec:preliminaries} presents basic concepts, the main preliminary results needed to understand the technical content of this paper, and results related to the three common techniques used in latter sections. \Cref{sec:progressive-flows} presents our and $O(k^2|V|+|E|)$ time algorithm for MPC\footnote{In \Cref{sec:antichain-structure} we show that our algorithm implicitly maintains a structure of antichains.}. \Cref{sec:edge-thinning} presents the algorithm of \Cref{thm:edge-thinning}. Omitted proofs can be found in the Appendices.

\section{Preliminaries} \label{sec:preliminaries}

\subsection{Basics}\label{sec:basics}
A directed graph is a tuple $G = (V, E)$, where $V$ is a set of vertices and $E$ is a set of edges, $E\subseteq V^2$. For an edge $e=(u,v) \in E$, it is said that $e$ goes \emph{from} $u$ \emph{to} $v$, that $u$ and $v$ are \emph{neighbors}, and that $e$ is \emph{incident} to both $u$ and $v$. In particular, $u$ is an \emph{in-neighbor} of $v$, $v$ is an \emph{out-neighbor} of $u$, $e$ is an edge \emph{incoming} to $v$ and \emph{outgoing} from $u$. We denote $N^+(v)$ ($N^-(v)$) to the set of out-neighbors (in-neighbors) of $v$, and by $I^+(v)$ ($I^-(v)$) the edges outgoing (incoming) from (to) $v$. A graph $S = (V_S, E_S)$ is said to be a \emph{subgraph} of $G$ if $V_S \subseteq V$ and $E_S \subseteq E$. If $V_S = V$ it is called \emph{spanning} subgraph. If $V' \subseteq V$, then $G[V']$ is the subgraph of $G$ \emph{induced} by $V'$, defined as $G[V'] = (V', E_{V'})$, where $E_{V'} = \{(u, v) \in E ~:~ u,v \in V'\}$. A \emph{path} $P$ in $G$ is a sequence of vertices $v_1, \ldots , v_{\ell}$ of $G$, such that $(v_i, v_{i+1}) \in E$, for all $i \in [1\ldots \ell-1]$, and $v_{i}\not= v_{j}$, for all $i\not=j$. For every $i,j \in\{1,\ldots,\ell\}, i\le j,$ $v_{i},\ldots,v_j$ is a \emph{subpath} of $P$. If $v_{1} = v_{\ell}$ it is called \emph{cycle}, and we denote it by $C$. If $\ell = |P|\ge 2$ it is said that the path is \emph{proper}. A \emph{directed acyclic graph} (DAG) is a directed graph without proper cycles. A \emph{topological ordering} of a DAG is a total order of $V$, $v_1,\ldots, v_{|V|}$, such that for all $(v_{i}, v_{j}) \in E$, $i < j$. A topological ordering can be computed in $O(|V|+|E|)$ time~\cite{kahn1962topological,tarjan1976edge}. If there exists a path $P = v_1, \ldots , v_{\ell}$ in $G$, with $u = v_1$ and $v = v_{\ell}$, it is said that $u$ \emph{reaches} $v$. A \emph{path cover} $\pathcover$ of $G$ is a set of paths such that every vertex $v\in V$ appears in some path of $\pathcover$. The \emph{size} $|\pathcover|$ is the number of paths of $\pathcover$, and its \emph{length} $||\pathcover||$ is the sum of the lengths of its paths, that is, $\sum_{P\in\pathcover} |P|$. If $\pathcover$ has minimum size among all path covers, then it is a \emph{minimum path cover} (MPC), and its size corresponds to the \emph{width} of $G$, that is, $\width(G) = \min_{\pathcover, \text{path cover}} |\pathcover|$ . An \emph{antichain} $A$ is a set of vertices such that for each $u,v \in A$ $u\not=v$ $u$, does not reach $v$, a \emph{maximum antichain} is an antichain of maximum size. Dilworth's theorem~\cite{dilworth2009decomposition} states that the size of a maximum antichain equals the size of an MPC. The \emph{multiplicity} of an edge $e\in E$ with respect to a set of paths $\pathcover$, $\mu_{\pathcover}(e)$ (only $\mu(e)$ if $\pathcover$ is clear from the context), is defined as the number of paths in $\pathcover$ that contain $e$, $\mu_{\pathcover}(e) = |\{P\in\pathcover\mid e \in P\}|$.

In our algorithm we work with subgraphs induced by a consecutive subsequence of vertices in a topological ordering. As such, the following lemma, proven by Cáceres et.al~\cite{caceres2021a}, shows that we can bound the width of these subgraphs by $k = \width(G)$.

\begin{restatable}[\cite{caceres2021a}]{lemma}{topologicalDoNotIncreaseWidth}
    \label{topologicalDoNotIncreaseWidth}
    Let $G = (V, E)$ be a DAG, and $v_1, \ldots , v_{|V|}$ a topological ordering of its vertices. Then, for all $i, j \in [1\ldots |V|], i \le j$, $\width(G_{i,j}) \le \width(G)$, with $G_{i,j} := G[\{v_{i}, \ldots , v_{j}\}]$.
\end{restatable}

\subsection{Minimum Flow}\label{sec:min-flow}
The problem of minimum flow with lower and upper bounds on edges has been studied before (see for example \cite{ahujia1993network,ciurea2004sequential,bang2008digraphs}). The concept of maximum ow-cuts has been studied before but only in the context of some specific problem solved by a reduction to minimum flow (see for example \cite{mohring1985algorithmic,pijls2013another,marchal2018parallel}). For completeness, in this section we include a proof for the case when only lower bounds on the edges are considered. The proof shown is an adaptation of the proof of the maximum flow/minimum cut theorem given in~\cite{williamson2019network}.

Given a (directed) graph $G = (V, E)$, a source $s \in V$, a sink $t \in V$, and a function of \emph{lower bounds} or \emph{demands} on its edges $d: E \rightarrow \Nzero$, an $st$-\emph{flow} (or just \emph{flow} when $s$ and $t$ are clear from the context) is a function on the edges $f: E \rightarrow \Nzero$, satisfying $f(e) \ge d(e)$ for all $e \in E$ ($f$ \emph{satisfies the demands}) and $\sum_{e \in I^-(v)} f(e) = \sum_{e \in I^+(v)} f(e)$ for all $v \in V \setminus \{s,t\}$ (\emph{flow conservation}). If a flow exists, the tuple $(G, s, t, d)$ is said to be a \emph{flow network}. The \emph{size} of $f$ is the net amount of flow exiting $s$, formally $|f| = \sum_{e \in I^+(s)} f(e) - \sum_{e \in I^-(s)} f(e)$. An $st$-\emph{cut} (or just \emph{cut} when $s$ and $t$ are clear from the context) is a partition $(S , T)$ of $V$ such that $s \in S$ and $t \in T$. An edge $(u, v)$ \emph{crosses} the cut $(S, T)$ if $u\in S$ and $v \in T$, or vice versa. If there are no edges \emph{crossing} the cut from $T$ to $S$, that is, if $\{(u,v) \in E \mid u \in T, v \in S\} = \emptyset$, then $(S, T)$ is a \emph{one-way cut} (ow-cut). The \emph{demand} of an ow-cut is the sum of the demands of the edges crossing the cut, formally $d((S,T)) = \sum_{e = (u, v), u \in S, v \in T} d(e)$, and we also denote it $d(S)$. An ow-cut whose demand is maximum among the demands of all ow-cuts is a \emph{maximum ow-cut}.

From these definitions the following properties can be derived:
\begin{basicprop}\label{prop:basicsMinFlow}
For a flow network $(G, s, t, d)$:
\begin{enumerate}
    \item[(a)] For any cut $(S,T)$ and flow $f$:
    \begin{align*}
        |f| = f(S) := \sum_{e = (u, v) \in E, u \in S, v \in T} f(e) - \sum_{e = (v, u) \in E, u \in S, v \in T} f(e).
    \end{align*}
    \item[(b)] For any ow-cut $(S,T)$ and flow $f$, $|f| \ge d((S,T))$.
    
\end{enumerate}
\end{basicprop}
\begin{proof}
  \begin{enumerate}
    \item[(a)] By definition of size, flow conservation and the fact that $(S, T)$ is a partition of $V$.
    \begin{align*}
        |f| &= \sum_{e \in I^+(s)} f(e) - \sum_{e \in I^-(s)} f(e)\\
        &= \left(\sum_{e \in I^+(s)} f(e) - \sum_{e \in I^-(s)} f(e)\right) + \sum_{u \in S\setminus \{s\}}\underbrace{\left(\sum_{e \in I^+(u)} f(e) - \sum_{e \in I^-(u)} f(e)\right)}_{=0}\\
        &=\sum_{u \in S}\left(\sum_{e \in I^+(u)} f(e) - \sum_{e \in I^-(u)} f(e)\right)\\
        &= \sum_{u \in S}\left(\sum_{e = (u, u') \in E, u' \in S} f(e) + \sum_{e = (u, v) \in E, v \in T} f(e)   - \sum_{e = (u', u) \in E, u' \in S} f(e) - \sum_{e = (v, u) \in E, v \in T} f(e)\right)\\
        %&= \sum_{u \in S}\left(\sum_{e = (u, v) \in E, v \in T} f(e)  - \sum_{e = (v, u) \in E, v \in T} f(e)\right) + \underbrace{\sum_{u \in S}\left(\sum_{e = (u, u') \in E, u' \in S} f(e) - \sum_{e = (u', u) \in E, u' \in S} f(e)\right)}_{=0}\\
        &= \sum_{u \in S}\left(\sum_{e = (u, v) \in E, v \in T} f(e)   - \sum_{e = (v, u) \in E, v \in T} f(e)\right)\\
        &= \sum_{e = (u, v) \in E, u \in S, v \in T} f(e) - \sum_{e = (v, u) \in E, u \in S, v \in T} f(e)
    \end{align*}
    
    \item[(b)] By using the previous property, the fact that ow-cuts do not have edges crossing from $T$ to $S$ and the lower bounds on the edges.
    \begin{align*}
        |f| &= \sum_{e = (u, v) \in E, u \in S, v \in T} f(e) - \underbrace{\sum_{e = (v, u) \in E, u \in S, v \in T} f(e)}_{=0} \\
        &= \sum_{e = (u, v) \in E, u \in S, v \in T} f(e)\\
        & \ge \sum_{e = (u, v) \in E, u \in S, v \in T} d(e)\\
        &= d((S,T))
    \end{align*}
    
\end{enumerate}  
\end{proof}

Given a \emph{flow network} $(G, s, t, d)$, the problem of \emph{minimum flow} consists of finding a flow $f^*$ of minimum size $|f^*|$ among the flows of the network, such a flow is a \emph{minimum flow}. If a minimum flow exists, then $(G, s, t, d)$ is a \emph{feasible} flow network. The following theorem relates the maximum demand of a ow-cut with the size of a minimum flow~\cite{ahujia1993network,ciurea2004sequential,bang2008digraphs}.

\begin{restatable}{thm}{maxowcutminflow}\label{thm:maxowcut-minflow}
Let $(G, s, t, d)$ be a feasible flow network. Then,
\begin{align*}
    \max_{(S, T), st\text{-}ow\text{-}cut} d((S,T)) = \min_{f, st\text{-}flow} |f|.
\end{align*}
\end{restatable}
\begin{proof}
Given a flow $f$ in $(G, s, t, d)$, the \emph{residual network} of $G$ with respect to $f$ is defined as $\residual{G}{f} = (V, E_f)$ with $E_f = \{(u, v) \mid (v, u) \in E\} \cup \{e \mid f(e) > d(e)\}$, that is, the \emph{reverse edges} of $G$, plus the edges of $G$ on which the flow can be decreased without violating the demands (\emph{direct edges})\footnote{In the literature of flows, a \emph{residual} capacity/demand is defined on the edges of $\residual{G}{f}$. We do not define these capacities/demands since our approaches to MPC do not exploit them.}. Note that a path from $s$ to $t$ in $\residual{G}{f}$ can be used to create another flow $f'$ of smaller size by increasing flow on reverse edges and decreasing flow on direct edges of the path, such a path its is called \emph{decrementing path}. Therefore, for a minimum flow $f^*$ there is no decrementing path in $\residual{G}{f^*}$. Taking $S$ as the set of vertices reachable from $s$ in $\residual{G}{f^*}$ (and $T = V\setminus S$), $(S,T)$ is an ow-cut ($s \in S$, $t \in T$, and there is no edge in $G$ from $T$ to $S$, since there is no edge in the opposite direction in $\residual{G}{f^*}$ by definition of $S$). Moreover, for every edge $e\in E$ from $S$ to $T$, $f(e) = d(e)$, since otherwise this edge would appear in $\residual{G}{f^*}$, which is not possible by definition of $S$. Therefore, the inequality of Property (b) is an equality and $|f^*| = d((S,T))$. Finally, since the demand of any ow-cut is a lower bound for the size of the flow, $(S,T)$ is a maximum ow-cut.
\end{proof}

\subsection{MPC in DAGs through Minimum Flow}\label{sec:minflow-reduction}
The reduction from MPC in DAGs to minimum flow has been stated several times in the literature~\cite{ntafos1979path,mohring1985algorithmic,gavril1987algorithms,Jagadish90,ciurea2004sequential,rademaker2012optimal,pijls2013another,marchal2018parallel,makinen2019sparse}, we include it here for completeness.

\begin{figure}[t]
    \centering
    \begin{subfigure}[b]{0.32\textwidth}
        \centering
        \includegraphics[width=0.95\textwidth]{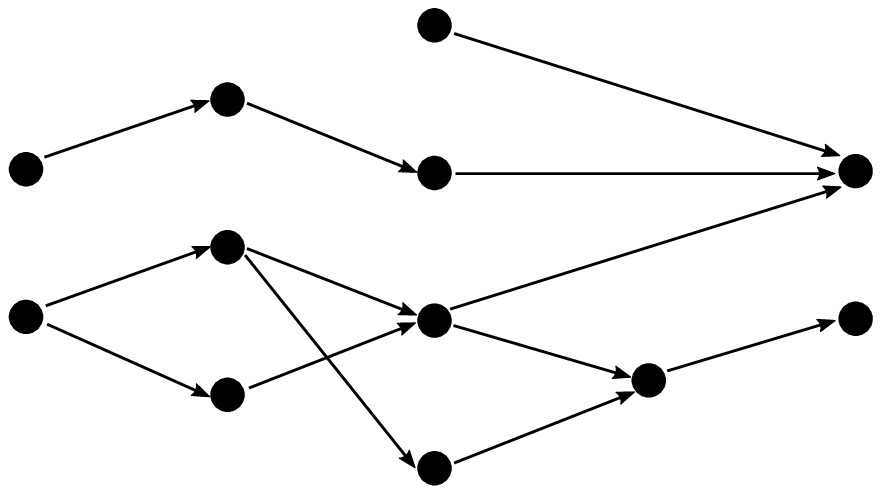}
        \caption[]%
        {{\small Example DAG}}
        \label{subfig:dag}
    \end{subfigure}
    \begin{subfigure}[b]{0.32\textwidth}
        \centering
        \includegraphics[width=0.95\textwidth]{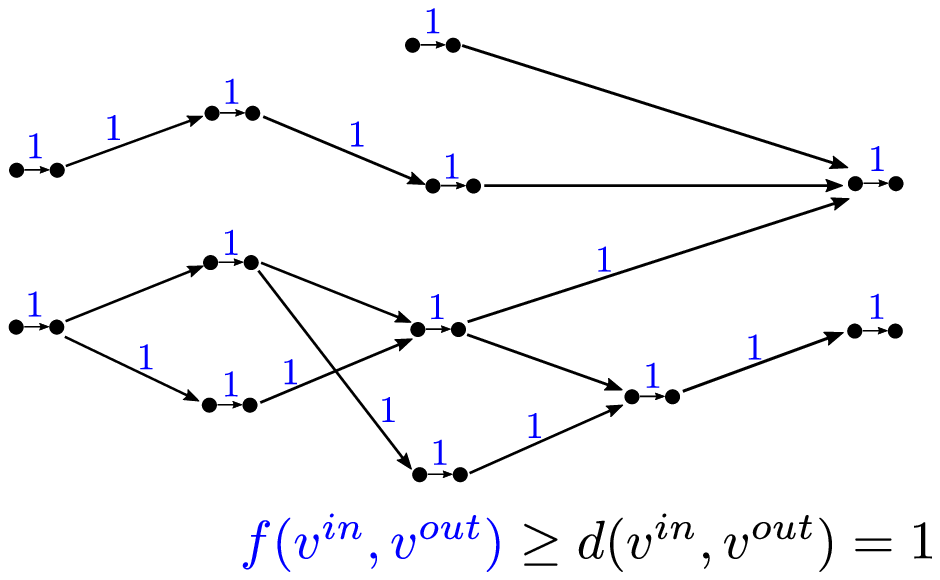}
        \caption[]%
        {{\small Flow reduction and min flow}}
        \label{subfig:minflow}
    \end{subfigure}
    \begin{subfigure}[b]{0.32\textwidth}
        \centering
        \includegraphics[width=0.95\textwidth]{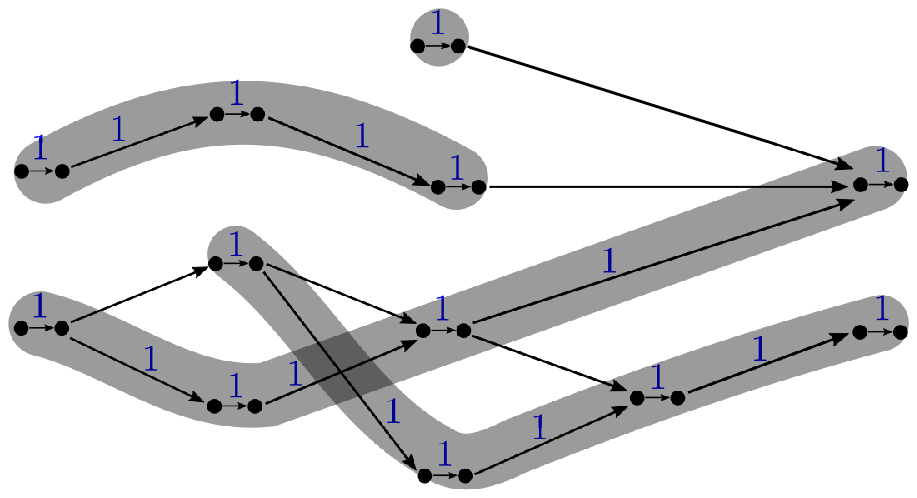}
        \caption[]%
        {{\small Flow decomposition}}
        \label{subfig:decomposition}
    \end{subfigure}
    \caption[]
    {\small Minimum Path Cover in DAGs reduced to Minimum Flow. \Cref{subfig:dag} shows an example DAG $G$. \Cref{subfig:minflow} shows the corresponding flow reduction $\flowG$ as well as a minimum flow $f: \flowE \to \mathbb{N}_0$. Vertices $s$ and $t$, and flow zero values were removed for simplicity. \Cref{subfig:decomposition} shows a flow decomposition of $f$ (each path highlighted in gray), which corresponds to an MPC of $G$.}
    \label{fig:minFlowReduction}
\end{figure}

The problem of finding an MPC in a DAG $G = (V, E)$ can be solved by a reduction to the problem of minimum flow on an appropriate feasible flow network $(\flowG = (\flowV, \flowE), s, t, d)$, defined as: $\flowV = \{s, t\} \cup \{v^{in} \mid v \in V\} \cup \{v^{out} \mid v \in V\}$ ($\{s,t\}\cap V = \emptyset$), that is, the source $s$, the sink $t$ and two vertices $v^{in}, v^{out}$ per a \emph{split} of every vertex $v \in V$; $\flowE = \{(s, v^{in}) \mid v \in V\} \cup \{(v^{out}, t) \mid v \in V\} \cup \{(v^{in}, v^{out}) \mid v \in V\} \cup \{(u^{out}, v^{in}) \mid (u,v) \in E\}$, that is, $s$ is connected to all vertices $v^{in}$, $t$ from all vertices $v^{out}$, the split vertices are connected from $v^{in}$ to $v^{out}$, and also the topology of $G$ is represented by connecting from $u^{out}$ to $v^{in}$ if $(u,v) \in E$. The demands are defined as $d(e) = 1$ if $e = (v^{in}, v^{out})$ for some $v \in V$ and $0$ otherwise. The tuple $(\flowG, s, t, d)$ is the \emph{flow reduction} of $G$. Note that $|\flowV| = 2|V| + 2 = O(|V|)$, $|\flowE| = 3|V| + |E| = O(|V| + |E|)$, and $\flowG$ is a DAG. See \Cref{fig:minFlowReduction}.

A path cover $\pathcover = P_1, \ldots, P_\ell$ of $G$ directly translates into a flow $f$ for $\flowG, s, t, d$ of size $|f| = \ell$. Starting with a function $f(e) = 0, e \in \flowE$ and iteratively increasing it. For every path $P_i$, it suffices to attach $s$ and $t$ at the ends and to replace every $v \in P_i$ by $v^{in}, v^{out}$, then the flow through the edges of the resulting path is increased by $1$. Since the flow is increased through paths from $s$ to $t$ this procedure maintains the flow conservation constrains, furthermore, since $\pathcover$ is a path cover, the flow through every edge $(v^{in}, v^{out})$ is increased by at least $1$ for every $v \in V$, thus $f$ corresponds to a flow of size $|\pathcover|$. 

Moreover, every flow $f$ of $(\flowG, s, t, d)$ can be decomposed into $|f|$ paths corresponding to a path cover of $G$. Iteratively, starting from $f$, a path $P$ from $s$ to $t$ whose edges have positive flow is found, and then the flow on the edges of $P$ is decreased by $1$. By flow conservation, $P$ can be found while $|f| > 0$, and since $|f|$ is decreased by $1$ at each iteration, exactly $|f|$ paths are obtained. By construction of $\flowG$ these paths can easily be transformed into a path cover of size $\ell$ of $G$, by removing $s$ and $t$ and merging the split vertices. Such a decomposition can be computed in time proportional to the total length of the MPC $\pathcover$, $O(||\pathcover||)$(see e.g.~\cite[Lemma 1.11 (full version)]{kogan2022beating}).

As such, a minimum flow of $(\flowG, s, t, d)$ provides an MPC of $G$. Moreover, by noting that a vertex cannot belong to more than $|V|$ paths of an MPC, one can add capacity $|V|$ to every edge, which allows to solve the problem by a well known reduction to \emph{maximum flow}~\cite[Theorem 3.9.1]{bang2008digraphs}.

\begin{restatable}{thm}{reductiontomaxflow}\label{thm:reduction-to-max-flow}
Given a DAG $G = (V,E)$, we compute an MPC $\pathcover$ of $G$ in time\\ $O(\texttt{MaxFlow}(|V|, |E|, |V|) + ||\pathcover||)$, where $\texttt{MaxFlow}(n, m, c)$ is the running time of a maximum flow algorithm on a network of $n$ vertices, $m$ edges and maximum capacity $c$.
\end{restatable}

The set of edges of the form $(v^{in}, v^{out})$ crossing a maximum ow-cut corresponds to a maximum antichain of $G$ (by merging the edges $(v^{in}, v^{out})$ into $v$, see~\Cref{sec:antichain-structure}). By further noting that $S = \{v \in \flowV \mid s\text{ reaches } v\text{ in }\residual{\flowG}{f^*}\}$ is a maximum ow-cut when $f^*$ is a minimum flow of $(\flowG, s, t, d)$, we obtain the following result.

\begin{restatable}{lemma}{fastantichainminflow}
\label{thm:main-antichain-minflow}
Given the flow reduction $(\flowG = (\flowV, \flowE), s, t, d)$ of a DAG $G = (V,E)$ and a minimum flow $f^*:\flowE: \rightarrow \mathbb{N}_0$ of it, we compute a maximum antichain of $G$ in time $O(|V|+|E|)$.
\end{restatable}
% \begin{proof}
%     We traverse $\residual{\flowG}{f^*}$ and find all vertices reachable from $s$, those vertices form a maximum ow-cut, the edges of the form $(v^{in}, v^{out})$ crossing this cut represent a maximum antichain of $G$ (see \Cref{sec:antichain-structure}).
% \end{proof}

As such, this derives a minimum (maximum) flow running time algorithm to compute a maximum antichain of a DAG.

\begin{restatable}{corollary}{fastantichainminflowcor}
\label{cor:main-antichain-minflow}
Given a DAG $G = (V,E)$, we compute a maximum antichain of $G$ in time $O(\texttt{MaxFlow}(|V|, |E|, |V|))$, where $\texttt{MaxFlow}(n, m, c)$ is the running time of a maximum flow algorithm on a network of $n$ vertices, $m$ edges and maximum capacity $c$.
\end{restatable}

More generally, every MPC algorithm (despite of its approach) can be used to obtain a maximum antichain, by first computing the corresponding residual.

\begin{restatable}{lemma}{fastantichain}
\label{thm:main-antichain}
Given a DAG $G = (V,E)$ and an MPC $\pathcover$, we compute a maximum antichain of $G$ in time $O(|V|+|E| + ||\pathcover||)$.
\end{restatable}
% \begin{proof}
%     We build the flow reduction $(\flowG, s, t, d)$ of $G$, and $\residual{\flowG}{\pathcover}$ in time $O(|V|+|E|+||\pathcover||)$ and apply \Cref{thm:main-antichain-minflow}.
% \end{proof}

Since $||\pathcover|| = O(k|V|)$, by combining the previous result with our MPC algorithm we obtain.

\begin{restatable}{corollary}{fastantichaincor}
\label{cor:main-antichain}
Given a DAG $G = (V,E)$, we compute a maximum antichain of $G$ in time $O(k^2|V|+|E|)$.
\end{restatable}

\subsection{Transitive sparsification} We say that a spanning subgraph $S = (V, E_S)$ of a DAG $G = (V, E)$ is a \emph{transitive sparsification} of $G$, if for every $u, v \in V$, $u$ reaches $v$ in $S$ if and only if $u$ reaches $v$ in $G$. Since $G$ and $S$ have the same reachability relations on their vertices, they share their antichains, thus $\width(G) = \width(S)$. As such, an MPC of $S$ is also an MPC of $G$, and the edges $E\setminus E_S$ can be safely removed for the purpose of computing an MPC of $G$. If we have a path cover $\pathcover$ of size $t$ of $G$, then we can \emph{sparsify} (remove some transitive edges) the incoming edges of a particular vertex $v$ to at most $t$ in time $O(t+|N^-(v)|)$. If $v$ has more than $t$ in-neighbors then two of them belong to the same path, and we can remove the edge from the in-neighbor appearing first in the path as this is transitive. To efficiently filter out those transitive edges we create an array of $t$ elements initialized as $survivor = (v_{-\infty})^t$, where $v_{-\infty} \not \in V$ is before every $v\in V$ in topological order. Then, we process the edges $(u,v)$ incoming to $v$, we set $i \gets path(u)$ ($path(u)$ gives the ID of some path of $\pathcover$ containing $u$) and if $survivor[i]$ is before $u$ in topological order we replace it $survivor[i] \gets u$. Finally, the edges in the sparsification are $\{(survivor[i], v) \mid i\in\{1,\ldots,t\} \land survivor[i] \not = v_{-\infty}\}$.

\begin{restatable}{obs}{obssparsification}
\label{obs:sparsificationVertex}
Let $G = (V, E)$ be a DAG, $\pathcover$ a path cover, $|\pathcover|=t$, $v$ a vertex of $G$, and $path:V\rightarrow \{1,\ldots,t\}$ a function  that answers in constant time $path(v)$, the ID of some path of $\pathcover$ containing $v$. We can sparsify the incoming edges of $v$ to at most $t$ in time $O(t+|N^-(v)|)$.
\end{restatable}

A \emph{path} function can be computed by scanning the path cover $\pathcover$ in time $O(||\pathcover||)$. If then we apply \Cref{obs:sparsificationVertex} to every vertex we obtain.

\begin{restatable}{lemma}{sparsificationalgorithm}
    \label{sparsification_algorithm}
    Let $G = (V, E)$ be a DAG, and $\pathcover$, $|\pathcover|=t$, be a path cover of $G$. Then, we can sparsify $G$ to $S = (V, E_S)$, such that $\pathcover$ is a path cover of $S$ and $|E_S| \le t|V|$, in $O(|V| + |E| = ||\pathcover||) = O(t|V|+|E|)$ time.
\end{restatable}
% \begin{proof}
%     Let $\pathcover = P_1,\ldots,P_t$. First, we traverse each path in time $O(t|V|)$ and compute for every vertex $path(v)$, which is the ID of some path containing $v$. We also initialize $v.survivor[i] = u$ if $(u, v)$ is an edge of path $P_i$ and $v_{-\infty}$ if such edge does not exist ($v_{-\infty}\not \in V$, is before every $v \in V$ in topological order). Then, we process the edges $e = (u, v)$ in time $O(|E|)$, set $i = path(u)$, and if $v.survivor[i]$ is before $u$ in topological order, we set $v.survivor[i] = u$. Finally, $E_S$ will be the edges $(v.survivor[i], v)$ such that $v.survivor[i] \not = v_{-\infty}$, thus there are at most $t|V|$. Note that $S$ contains all the edges in the paths because we initialized $v.survivor[i] = u$ for every edge $(u, v)$ in path $P_i$, and these are not updated during the algorithm, thus $\pathcover$ is also a path cover of $S$. Now we prove that $S$ is a transitive sparsification of $G$. If an edge appears in $S$, is of the form $(u = v.survivor[i], v)$ for some edge $(u, v)$ of $G$, thus $S$ is a subgraph of $G$. Finally, if an edge $(u, v)$ is not considered in $S$ it means that there is an edge $(x, v)$ such that, $u, x \in P_i$ with $u$ before $x$ in $P_i$. Therefore, there is a path from $u$ to $v$ using the corresponding edges of $P_i$ followed by $(x, v)$.
% \end{proof}

The following lemma shows that we can locally sparsify a subgraph and apply these changes to the original graph to obtain a transitive sparsification. In other words, it allows us to remove transitive edges when working locally in a subgraph of $G$.
\begin{restatable}{lemma}{sparsificationOfSubgraph}
    \label{sparsificationOfSubgraph}
    Let $G = (V, E)$ be a graph, $S = (V_S, E_S)$ a subgraph of $G$, and $S' = (V_S, E_{S'})$ a transitive sparsification of $S$. Then $G' = (V, E\setminus (E_S \setminus E_{S'}))$ is a  transitive sparsification of $G$.
\end{restatable}
\begin{proof}
    Since $S'$ is a transitive sparsification of $S$, $E_{S'}\subseteq E_S$ thus $E\setminus (E_S \setminus E_{S'}) \subseteq E$ and then $G'$ is a subgraph of $G$. Now, suppose by contradiction that $u$ and $v$ are connected in $G$ by a path $P$, but they are not connected in $G'$. Then, $P$ contains an edge $e = (a, b)\in E_S\setminus E_{S'}$ disconnecting $b$ from $a$ in $S'$, but since $S'$ is a transitive sparsification of $S$, $a$ is connected to $b$ in $S'$, which is a contradiction.
\end{proof}

\subsection{Shrinking} As explained before, shrinking is the process of transforming an arbitrary path cover $\pathcover$ into an MPC, and it can be solved by finding $|\pathcover|-\width(G)$ decrementing paths in the residual of the flow induced by $\pathcover$ in $\flowG$, and then decomposing the resulting flow into an MPC.  M\"akinen et al.~\cite{makinen2019sparse} apply this idea to shrink a path cover of size $O(k\log{|V|})$. We generalize this approach in the following lemma.

\begin{restatable}{lemma}{shrinkingalgorithm}
\label{result:shrinking}
    Given a DAG $G = (V, E)$ of width $k$, and a path cover $\pathcover$, $|\pathcover|=t$, of $G$, we can obtain an MPC of $G$ in time $O(t(|V| + |E|))$.
\end{restatable}
\begin{proof}
    We build the flow reduction $(\flowG, s, t, d)$ of $G$, and $\residual{\flowG}{f}$ (where $f$ is the flow induced by $\pathcover$ in $\flowG$) in time $O(|V|+|E|+||\pathcover||) = O(t|V|+|E|)$. Then, we shrink the corresponding flow to minimum by finding $\le t-k$ decrementing paths in $\le t-k$ traversals of $\flowG$, and finally, we decompose the minimum flow into an MPC in additional $k$ traversals (one per path) of $\flowG$. In total this takes $O(t(|V|+|E|))$ time.
\end{proof}

\section{A parameterized linear time MPC algorithm}\label{sec:progressive-flows}
In this section we describe our algorithm to find an MPC of a $k$-width DAG $G = (V, E)$ in time $O(k^2|V|+|E|)$ (\Cref{thm:improved-mpc-flow}). We start by describing a simpler version of the algorithm running in time $O(k^3|V|+|E|)$\footnote{We further simplify the result presented in SODA22~\cite{caceres2022sparsifying} by getting rid of \emph{splicing} in the algorithm.}, which shows the main ideas behind our approach. After that, we show how to reduce the running time dependency on the width to $k^2$. Finally, we show that our algorithm implicitly maintains a structure of decreasing antichains sweeping the DAG during the algorithm.

\subsection{Overview of the algorithm}

We rely on the reduction from MPC in a DAG to minimum flow (see \Cref{sec:preliminaries}). We process the vertices of $G$ one by one in a topological ordering $v_1, \ldots, v_{|V|}$. At each step, we maintain a minimum flow $f^*_i$ of $\flowG_i = (\flowV_i, \flowE_i)$ (the flow reduction of a transitive sparsification of $G_i = G[\{v_1, \ldots, v_i\}]$), and a flow decomposition $\pathcover_i$ of $f^*_i$, that is, an MPC of $G_i$. When the next vertex $v_{i+1}$ is considered, we first use $\pathcover_{i}$ to sparsify its incoming edges to at most $|f^*_{i}| \le k$ in time $O(k+|N^-(v_{i+1})|)$ (see \Cref{obs:sparsificationVertex} and \Cref{topologicalDoNotIncreaseWidth,sparsificationOfSubgraph}). Then, we set $\tmp_{i+1} \gets \pathcover_i \cup \{(v_{i+1}^{in}, v_{i+1}^{out})\}$, where $(v_{i+1}^{in}, v_{i+1}^{out})$ corresponds to the edge representing $v_{i+1}$ in the flow reduction (we represent $st$-flow paths either as a sequence of vertices or edges excluding the extremes for convenience), and $t_{i+1}$ as the flow induced by $\tmp_{i+1}$ in $\flowG_{i+1}$. After that, we search for a decrementing path in $\residual{\flowG_{i+1}}{t_{i+1}}$. If the search fails we set $\pathcover_{i+1} \gets \tmp_{i+1}, f^*_{i+1} \gets t_{i+1}$. Otherwise, we obtain $f^*_{i+1}$ by modifying $f^*_i$ according to the decrementing path and decompose $f^*_{i+1}$ into $\pathcover_{i+1}$. Once all vertices have been processed $f^*_{|V|}$ will be a minimum flow of $\flowG$ and $\pathcover_{|V|}$ an MPC of $\flowG$.

Note that by processing the vertices in topological order we have that either $|\pathcover_{i+1}| = |\pathcover_{i}| + 1$ or $|\pathcover_{i+1}| = |\pathcover_{i}|$ (recall \Cref{topologicalDoNotIncreaseWidth}), and thus it suffices to traverse $\residual{\flowG_{i+1}}{t_{i+1}}$ once. In general, such a traversal for a decrementing path can be performed in $O(|V|+|E|)$ time, leading to a quadratic algorithm. However, our algorithm guides the search by assigning an integer level $\ell_i(v)$ to each vertex $v$ in $\flowG_{i}$. The search is performed in a \emph{layered} manner: it starts from the highest reachable layer (the vertices of highest level according to $\ell_i$), and it only continues to the next highest reachable layer once all reachable vertices from the current layer have been visited (see \Cref{sec:layered-traversal}). The definition of the level assignment is algorithmic (see \Cref{sec:level-updates}) and driven by the maintenance of three structural invariants (see \Cref{sec:levels-and-invariants}), which in turn allows an efficient implementation of the layered traversal. We formally define these concepts next.

\subsection{Levels, layers and invariants}\label{sec:levels-and-invariants}
We define the level assignment given to the vertices of $\flowG_{i}$, $\ell_i: \flowV_i \to \{0, 1, \ldots, M_i\} \cup \{-\infty, +\infty\}$, and the invariants maintained on $\ell_i$. The values $-\infty,\infty$ are reserved exclusively for $\ell_i(s) = -\infty, \ell_i(t) = \infty$. A \emph{layer} is a maximal set of vertices with the same level, thus layer $l$ is $L_i^l = \{v \in V(\flowG_{i}) \mid \ell_i(v) = l\}$. All layers form a partition of $\flowV_i = \bigcupdot_l L_i^l$. We define the $l$-th \emph{layered cut} as the union of layers from $-\infty$ to $l$, $L_i^{\le l} = \bigcupdot_{j=-\infty}^l L_i^j$, and its complement $L_i^{> l} = \bigcupdot_{j=l+1}^\infty L_i^j$. We are ready to present the three structural invariants maintained by our algorithm.\vspace{0.5em}

\begin{description}
    \item[Invariant A]: If $(u, v)$ is an edge in $\residual{\flowG_i}{f^*_i}$ and $\{u, v\} \cap \{s, t\} = \emptyset$, then $\ell_i(u) \ge \ell_i(v)$. 
    \item[Invariant B]: \begin{itemize}
        \item If $f^*_i(u^{out}, t) > 0$, then $f^*_i(u^{out}, t) = 1$ and $\ell_i(u^{in}) < \ell_i(u^{out})$.
        \item If $f^*_i(s, u^{in}) > 0$, then $\ell_i(u^{in}) = 0$.
    \end{itemize} 
    \item[Invariant C]: If $l, l'\in \{0, 1, \ldots, M_i-1\}$ with $l < l'$, then $d(L_i^{\le l}) > d(L_i^{\le l'})$. 
\end{description}
\vspace{0.5em}

Note that \invA{} implies that the layered cuts are ow-cuts, since otherwise the corresponding reverse edge would appear in the residual breaking the invariant (and there are no edges from $t$). Moreover, the first layered cut is a maximum ow-cut.

\begin{lemma}\label{lemma:layeredCutMaxCut}
The first layered cut, $L^{\le 0}_i$, is a maximum ow-cut.
\end{lemma}
\begin{proof}
Indeed, by \invB{}, every edge $(u, v)$ crossing the ow-cut $L^{\le 0}_i$ with $f^*_i(u, v) > 0$ satisfies $\{u, v\} \cap \{s, t\} = \emptyset$. As such, by \invA{}, $f^*_i(u, v) = d(u,v)$, since otherwise the edge $(u, v)$ would appear in the residual breaking the invariant. Finally, we conclude that $f^*_i(L^{\le 0}_i) = |f^*_i| = d(L^{\le 0}_i)$ and $L^{\le 0}_i$ is a maximum ow-cut.
\end{proof}

% We extend the definition of level assignment to paths, the level of a path is the maximum level of a vertex in the path, that is, if $P$ is a path of $\flowG_i$, then $\ell(P) = \max_{v \in P} \ell(v)$. We define $\pathcover_i^{\ge l} \subseteq \pathcover_i$, as the flow paths whose level is at least $l$, $\pathcover_i^{\ge l} = \{P \in \pathcover_i \mid \ell(P)\ge l\}$. Note that $|\pathcover_i^{\ge l}| \ge |\pathcover_i^{\ge l'}|$ if $l' > l$. 

% At the beginning we fix $\ell(s) = -\infty$ and $\ell(t) =  +\infty$. We also maintain that $0\le\ell(v)\le\width(G_i)$ for all $v \in \flowV_i \setminus \{s,t\}$. Additionally, we maintain the following invariants:
% \begin{description}
%     \item[Invariant A]: If $(u, v)$ is an edge in $\residual{\flowG_i}{\pathcover_i}$ and $\{u, v\} \cap \{s, t\} = \emptyset$, then $\ell(u) \ge \ell(v)$. 
%     \item[Invariant B]: If $(u^{in}, u^{out})$ is the last edge of some $P \in \pathcover_i$, then $\ell(u^{in}) < \ell(u^{out})$.
%     \item[Invariant C]: If $l, l'$ are positive integers with $l' > l$, then $|\pathcover_i^{\ge l}| > |\pathcover_i^{\ge l'}|$. 
% \end{description}

% Note that, since we do not include $s$ and $t$ in the representation of flow paths, $0\le\ell(P)\le\width(G_i)$ for all $P \in \pathcover_{i}$, moreover, by \invB{}, $\ell(P) \ge 1$, thus $\pathcover_i^{\ge 1} = \pathcover_i$. Also note that \invC{} implies that every layer $l \in \{1,\ldots,L\}, L =\max_{v\in\flowV_i\setminus\{t\}}\ell(v)$ is not empty.

Finally, note that, by \invC{} and \Cref{lemma:layeredCutMaxCut}, $M_i \le |f^*_i| \le k$.

\subsection{The algorithm}\label{sec:progressive-flows-algorithm}

\subsubsection{Initial sparisifcation}\label{sec:initial-sparsification}
Our algorithm starts by using $\pathcover_i$ to obtain at most $|f^*_i|$ edges incoming to $v_{i+1}$ in time $O(|f^*_i|+|N^-(v_{i+1})|) = O(k+|N^-(v_{i+1})|)$ (see \Cref{obs:sparsificationVertex}). This procedure requires to answer $path(v)$ (the ID of some path of $\pathcover_{i}$ containing $v$) queries in constant time. To satisfy this requirement, we maintain path IDs on every vertex/edge of every flow path $P \in \pathcover_i$. The following lemma states that the sparsification of incoming edges in $G_{i+1}$ produces an sparsification of outgoing edges in the corresponding residual.

\begin{lemma}\label{lemma:sparsifiedresidual}
For every $x \in \flowV_{i+1}\setminus \{s, t\}$, $|N^+({x})| = O(|f^*_{i}|)$, in $\residual{\flowG_{i+1}}{t_{i+1}}$.
\end{lemma}
\begin{proof}
If $x$ is of the form $v^{in}$, then its only direct edge could be $(v^{in}, v^{out})$ (if $t_{i+1}(v^{in}, v^{out}) > 1$), its reverse edges are of the form $(v^{in}, u^{out})$, such that $(u, v)$ is an edge in $G_{i+1}$, thus there are at most $|f^*_i|$ such edges because of sparsification (recall that $|f^*_j| \le |f^*_i|$ for $j < i$, by \Cref{topologicalDoNotIncreaseWidth}). On the other hand, if $x$ is of the form $u^{out}$, then the only reverse edge is $(u^{out}, u^{in})$. To bound the number of direct edges consider the ow-cut $S =\{v \in \flowV_{i+1} : \mbox{$v$ reaches $u^{out}$ in $\flowG_{i+1}$}\}$, then $t_{i+1}(S) = |t_{i+1}| = |f^*_{i}| + 1$, and thus the number of direct edges $(u^{out}, v^{in})$ is at most $|f^*_i|+1$.
\end{proof}

%This lemma allows us to charge $O(|\pathcover_{i}|)$ to each visited vertex in the traversal for the search of a decrementing path.

\subsubsection{Layered traversal}\label{sec:layered-traversal}
Our layered traversal performs a BFS in each reachable layer from highest to lowest. If $t$ is reached, the search stops and the algorithm proceeds to obtain $f^*_{i+1}$ and $\pathcover_{i+1}$ from the decrementing path found. Since $f^*_{i}$ is a minimum flow of $\flowG_i$, every decrementing path $D$ in $\residual{\flowG_{i+1}}{t_{i+1}}$ starts with the edge $(s, v_{i+1}^{in})$ and ends with an edge of the form $(u^{out}, t)$ with $t_{i+1}(u^{out}, t) > 0$. Moreover, since $(v_{i+1}^{in}, v_{i+1}^{out})$ does not exist in $\residual{\flowG_{i+1}}{t_{i+1}}$\footnote{Recall that $t_{i+1}(v_{i+1}^{in}, v_{i+1}^{out}) = d(v_{i+1}^{in}, v_{i+1}^{out}) = 1$.}, the second edge of $D$ must be a reverse edge of the form $(v_{i+1}^{in}, u^{out})$, such that $u$ is an in-neighbor of $v_{i+1}$ in $G_{i+1}$.

We work with $M_{i}+1$ queues $Q_{0}, Q_1, \ldots, Q_{M_{i}}$ (one per layer), where $Q_j$ contains the \emph{enqueued} elements from $L^j_i$(layer $j$), therefore it is initialized as $Q_j \gets \{u^{out} \mid (u^{out}, v_{i+1}^{in}) \in \flowE_{i+1} \land \ell(u^{out}) = j\}$. By \Cref{lemma:sparsifiedresidual}, this initialization takes $O(|f^*_{i}| + M_i) = O(k)$ time. We start working with $Q_{M_i}$. When working with $Q_j$, we obtain the first element $u$ from the queue (if no such element exists we move to $L^{j-1}_i$ and work with $Q_{j-1}$), then we \emph{visit} $u$ and for each non-visited out-neighbor $v$ we add $v$ to $Q_{\ell(v)}$. Since edges in the residual do not increase the level (\invA), out-neighbors can only be added to queues at an equal or lower layer. As such, this traversal advances in a \emph{layered} manner, and it finds a decrementing path if one exists. Note that the running time of the layered traversal can be bounded by $O(|f^*_i|)=O(k)$ per visited vertex by \Cref{lemma:sparsifiedresidual}.

\begin{obs}\label{obs:end-vertices}
Let $E_{i} = \{u^{out} \in \flowV_i \mid f^*_{i}(u^{out}, t) > 0\}$, and $A_{i+1}$ the singleton set containing the last vertex $a_{i+1}^{out}$ in the decrementing path $D$ found by the layered traversal in $\residual{\flowG_{i+1}}{t_{i+1}}$, or the empty set if no decrementing path was found. Then, $E_{i+1} = E_i \cup \{v_{i+1}^{out}\} \setminus A_{i+1}$.
\end{obs}
\begin{proof}
    If no decrementing path was found the observation easily follows. On the other hand, if a decrementing path $D$ is found, the observation follows from the fact that the only edge in $D$ of the form $(u^{out}, t)$ with $u^{out} \in E_i$, comes from $a_{i+1}^{out}$, and by \invB{}, $f^*_{i}(a_{i+1}^{out}, t) = 1$, thus $f^*_{i+1}(a_{i+1}^{out}, t) = 0$ implying $a_{i+1}^{out}\not\in E_{i+1}$.
\end{proof}

 \subsubsection{Flow, level and path updates}\label{sec:level-updates}
Recall that if no decrementing path is found then we set $f^*_{i+1} \gets t_{i+1}$. Otherwise, we get $f^*_{i+1}$ by modifying $t_{i+1}$ according to the decrementing path $D$, as is standard in maximum/minimum flow algorithms, that is,
\begin{align*}
    f^*_{i+1}(u,v) = 
     \begin{cases}
        t_{i+1}(u,v) - 1 &\quad\text{if } (u, v) \in D \land (u, v) \in \flowE_{i+1} \\
        t_{i+1}(u,v) + 1 &\quad\text{if } (u, v) \in D \land (v, u) \in \flowE_{i+1} \\
       t_{i+1}(u,v)&\quad\text{otherwise}\\
     \end{cases}
\end{align*}

After obtaining $f^*_{i+1}$, we update the level of some vertices of $\flowV_{i+1}$ to maintain the invariants (\Cref{sec:levels-and-invariants}) on the new level assignment $\ell_{i+1}$. Moreover, to sparsify (\Cref{sec:progressive-flows-algorithm}) in the next iteration, we also compute the $path$ function (recall that $path(v)$ that gives the ID of some path of the MPC containing $v$) by decomposing $f^*_{i+1}$ into $\pathcover_{i+1}$.

If the smallest layer visited during the traversal is $L^l_i$, then we set $\ell_{i+1}(v_{i+1}^{in}) \gets l$, $\ell_{i+1}(v_{i+1}^{out}) \gets l+1$ (to maintain \invB{}), and change the level of every vertex $u$ visited during the traversal to $\ell_{i+1}(u) \gets l$ (to maintain \invA{}). Note that these level changes are made even if no decrementing path was found, in which case $l = 0$.

We note that after applying the previous level changes the following hold.

\begin{obs}\label{obs:last-vertices-dont-change-level}
$\ell_{i}(u^{out}) = \ell_{i+1}(u^{out})$ for every $u^{out} \in E_{i}\setminus A_{i+1}$
\end{obs}
\begin{proof}
The observation follows from the fact that the layered traversal only visits $a_{i+1}^{out}$ among all $E_{i}$ (if any), and thus only the level of $a_{i+1}^{out}$ changes to the smallest level visited.
\end{proof}

\begin{lemma}\label{lemma:invAB-before-merging}
\invAB{} hold for $\flowG_{i+1}, f^*_{i+1},\ell_{i+1}$.
\end{lemma}
\begin{proof}
For \invA{} consider $(u, v) \in \residual{\flowG_{i+1}}{f^*_{i+1}}$ with $\{u, v\} \cap \{s, t\} = \emptyset$, and assume for a contradiction that $\ell_{i+1}(u) < \ell_{i+1}(v)$. Since (inductively) \invA{} holds for $\flowG_{i}, f^*_{i},\ell_{i}$, then only $u$ was visited by the layered traversal. Indeed, if both were visited then $\ell_{i+1}(u) = \ell_{i+1}(v)$ a contradiction, and if none or only $v$ was visited then $\ell_{i+1}(u) = \ell_{i}(u) \ge \ell_{i}(v) \ge \ell_{i+1}(v)$ also a contradiction (the last inequality follows since visiting $v$ can only decrease its level). Note that, by construction of $\ell_{i+1}$, every in-neighbor $w$ of $v_{i+1}^{in}$ in $\residual{\flowG_{i+1}}{f^*_{i+1}}$ is such that $\ell_{i+1}(w) \ge \ell_{i+1}(v_{i+1}^{in})$, and thus $(u, v) \in \flowG_{i}$. Moreover, since only $u$ was visited $f^*_{i}(u,v) = f^*_{i+1}(u,v)$, and thus $(u, v) \in \residual{\flowG_{i}}{f^*_{i}}$. However, this contradicts the operation of the layered traversal, since the smallest level visited was $\ell_{i+1}(u)$, but vertex $v$ at level $\ell_{i}(v) = \ell_{i+1}(v)$ ($ > \ell_{i+1}(u)$) was not visited.

For the first part of \invB{} consider $f^*_{i+1}(u^{out}, t) > 0$, then by \Cref{obs:end-vertices}, $u^{out} \in E_{i+1} = E_{i} \cup \{v_{i+1}^{out}\} \setminus A_{i+1}$. If $u^{out} = v_{i+1}^{out}$, by construction of $\ell_{i+1}$, $\ell_{i+1}(v_{i+1}^{out}) = \ell_{i+1}(v_{i+1}^{in}) + 1$, and by construction of $f^*_{i+1}$, $f^*_{i+1}(v_{i+1}^{out}, t) = 1$. Otherwise, we have that $u^{out} \in E_{i} \setminus A_{i+1}$, by \Cref{obs:last-vertices-dont-change-level}, $\ell_{i+1}(u^{out}) = \ell_{i}(u^{out}) > \ell_{i}(u^{in}) \ge \ell_{i+1}(u^{in})$, and inductively $f^*_{i+1}(u^{out}, t) = f^*_{i}(u^{out}, t) = 1$.

Finally, for the second part of \invB{} consider $f^*_{i+1}(s, u^{in}) > 0$. If $u^{in} = v_{i+1}^{in}$, then no decrementing path was found and by construction of $\ell_{i+1}$, $\ell_{i+1}(v_{i+1}^{in}) = 0$. Otherwise, note that $f^*_{i+1}(s, u^{in}) = f^*_{i}(s, u^{in})$ but then $\ell_{i}(u^{in}) = 0 \ge \ell_{i+1}(u^{in}) = 0$.
\end{proof}

\begin{lemma}\label{lemma:level-path-changes}
If $l$ is the smallest level visited by the layered traversal of $\residual{\flowG_{i+1}}{t_{i+1}}$, then $d(L^{\le l'}_{i+1}) = d(L^{\le l'}_{i})$ for every $l' \in \{0,..,M_i\}\setminus\{l\}$, and $d(L^{\le l}_{i+1}) = d(L^{\le l}_{i}) + 1$.
\end{lemma}
\begin{proof}
By \invAB{} we have that
\begin{align*}
    f^*_i(L^{\le l'}_{i}) = |f^*_i| = d(L^{\le l'}_{i}) + \sum_{u^{out} \in L^{\le l'}_{i}} f^*_i(u^{out}, t) = d(L^{\le l'}_{i}) + |\{u^{out} \in L^{\le l'}_{i} \cap E_i\}|
\end{align*}
Moreover, by \Cref{lemma:invAB-before-merging}, defining $E_{i}^{\le l'} = E_{i} \cap L^{\le l'}_{i}$ and reordering, we have

\begin{align*}
    d(L^{\le l'}_{i}) &= |f^*_{i}| - |E_{i}^{\le l'}|\\
    d(L^{\le l'}_{i+1}) &= |f^*_{i+1}| - |E_{i+1}^{\le l'}|
\end{align*}

Consider $ d(L^{\le l'}_{i+1})- d(L^{\le l'}_{i})$, then by \Cref{obs:end-vertices,obs:last-vertices-dont-change-level}, we have
\begin{align*}
     d(L^{\le l'}_{i+1})- d(L^{\le l'}_{i}) &= |f^*_{i+1}| - |f^*_{i}| + |E_{i}^{\le l'}| - |(E_{i}^{\le l'} \cup \{v_{i+1}^{out}\} \setminus A_{i+1}\}) \cap L_{i+1}^{\le l'}|\\
     &= |f^*_{i+1}| - |f^*_{i}| + |E_{i}^{\le l'}| - (|E_{i}^{\le l'}| + |\{v_{i+1}^{out}\} \cap L_{i+1}^{\le l'}| - |A_{i+1}\cap L_{i+1}^{\le l'}|)\\
     &= |f^*_{i+1}| - |f^*_{i}| + |A_{i+1}\cap L_{i+1}^{\le l'}| - |\{v_{i+1}^{out}\} \cap L_{i+1}^{\le l'}|
\end{align*}

If a decrementing path was not found then $|f^*_{i+1}| = |f^*_{i}| + 1$ and $A_{i+1} = \emptyset$, thus
\begin{align*}
    d(L^{\le l'}_{i+1})- d(L^{\le l'}_{i}) &= 1 - |\{v_{i+1}^{out}\} \cap L_{i+1}^{\le l'}|
\end{align*}

Since $\ell_{i+1}(v_{i+1}^{out}) = l+1 = 1$, then $d(L^{\le l'}_{i+1})- d(L^{\le l'}_{i}) = 0, \forall l' \in \{0,\ldots,M_{i}\}\setminus\{l\}$, and $d(L^{\le 0}_{i+1})- d(L^{\le 0}_{i})= 1$. Otherwise (a decrementing path was found), $|f^*_{i+1}| = |f^*_{i}|$ and $A_{i+1} = \{a_{i+1}^{out}\}$, thus 
\begin{align*}
    d(L^{\le l'}_{i+1})- d(L^{\le l'}_{i}) &= |\{a_{i+1}^{out}\} \cap L_{i+1}^{\le l'}| - |\{v_{i+1}^{out}\} \cap L_{i+1}^{\le l'}|
\end{align*}

Since $\ell_{i+1}(v_{i+1}^{out}) = l+1 = \ell_{i+1}(a_{i+1}^{out}) + 1$, then $d(L^{\le l'}_{i+1})- d(L^{\le l'}_{i}) = 0, \forall l' \not= l$, and $d(L^{\le l}_{i+1})- d(L^{\le l}_{i}) = 1$.
\end{proof}

\Cref{lemma:level-path-changes} shows that \invC{} is also (inductively) maintained since the demand of the layered cuts remain the same as in the previous iteration, except $d(L^{\le l})$, which increases by one. If this increment breaks \invC{}, i.e. $d(L_{i+1}^{\le l}) = d(L_{i+1}^{\le l-1})$, we decrement the level of every vertex $u$, $\ell_{i+1}(u) \ge l$, by one. We call this procedure \emph{merge} of layer $l$ and runs in $O(1)$ time per vertex of level $l$ or more. Note that the merge of layer $l$ naturally repairs \invC{}. Moreover, \invAB{} are also maintained after the merge. Indeed, for \invA{} it suffices to note that if $\ell_{i+1}(u) = \ell_{i+1}(v)$, then both decrement their level by one. As for \invB{}, we note that (by construction) layer $0$ is never merged (second part), the merge does not modify the flow (first part flow condition), and (before the merge) there are no vertices $u^{out}$ with $f^*_{i+1}(u^{out}, t) > 0$ at layer $l$\footnote{The condition $d(L_{i+1}^{\le l}) = d(L_{i+1}^{\le l-1})$ that triggers the merge implies that $|E^{\le l-1}_{i+1}| = |E^{\le l}_{i+1}|$, see \Cref{lemma:level-path-changes}.} (first part level condition). \Cref{fig:progressive} illustrates the evolution of the level assignment in a step of the algorithm.

% one row version
\begin{figure}[t]
    \centering
    \begin{subfigure}[b]{0.32\textwidth}
        \centering
        \includegraphics[width=0.95\textwidth]{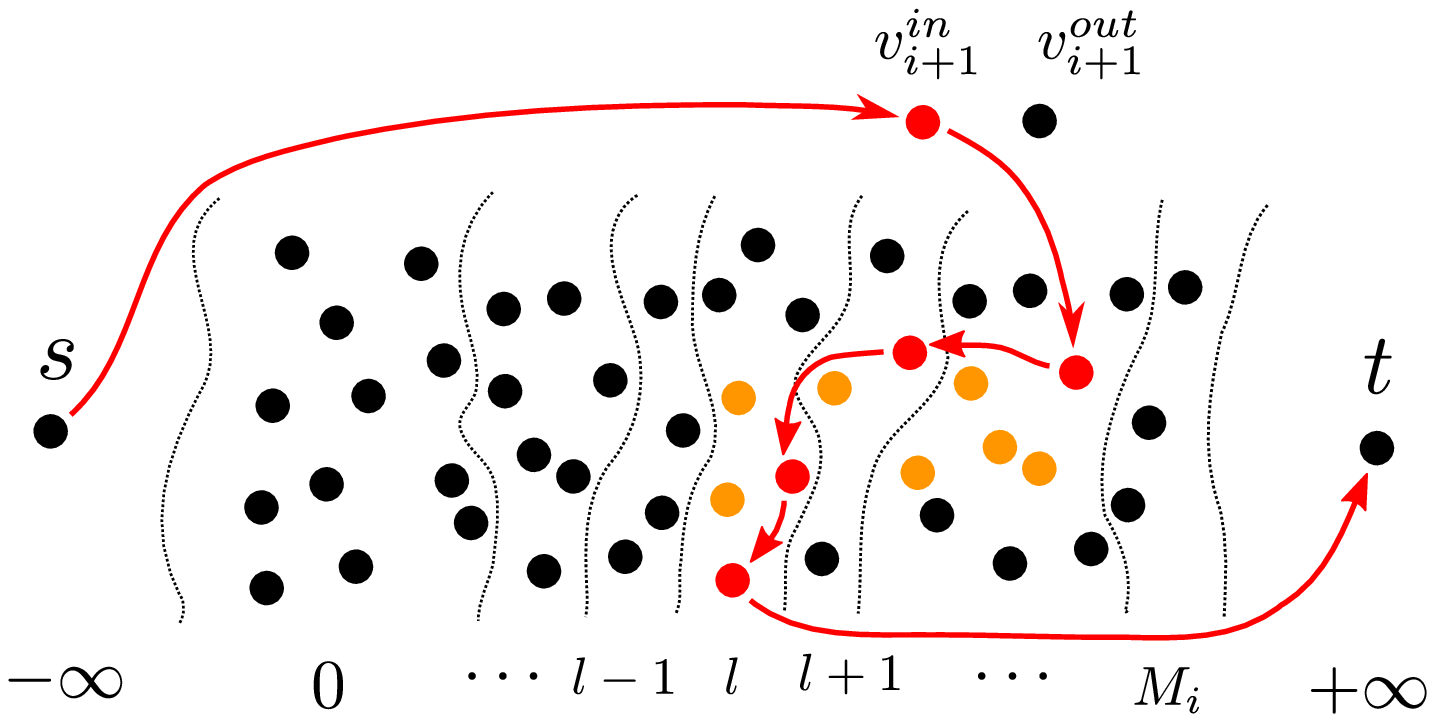}
        \caption[]%
        {{\small Layered traversal}}
        \label{subfig:traversal}
    \end{subfigure}
    \begin{subfigure}[b]{0.32\textwidth}
        \centering
        \includegraphics[width=0.95\textwidth]{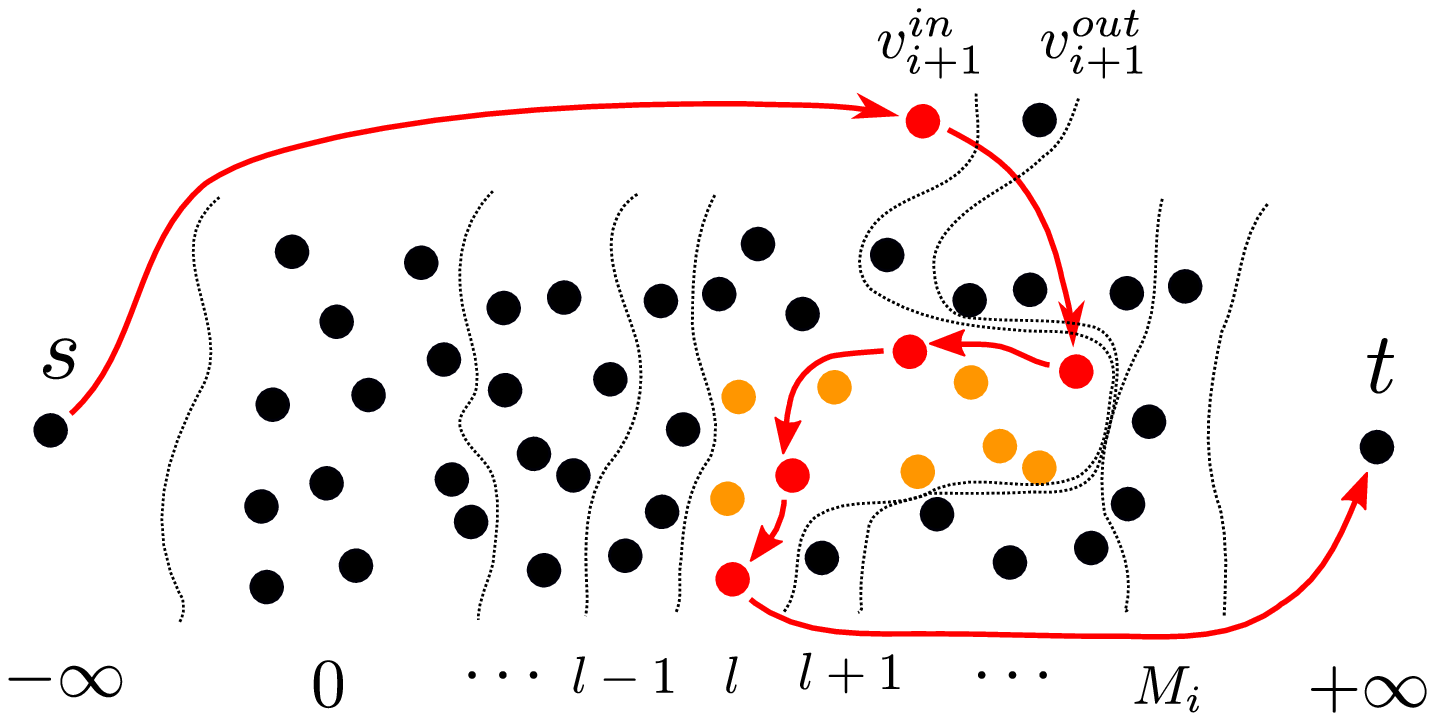}
        \caption[]%
        {{\small Level updates}}
        \label{subfig:updates}
    \end{subfigure}
    \begin{subfigure}[b]{0.32\textwidth}
        \centering
        \includegraphics[width=0.95\textwidth]{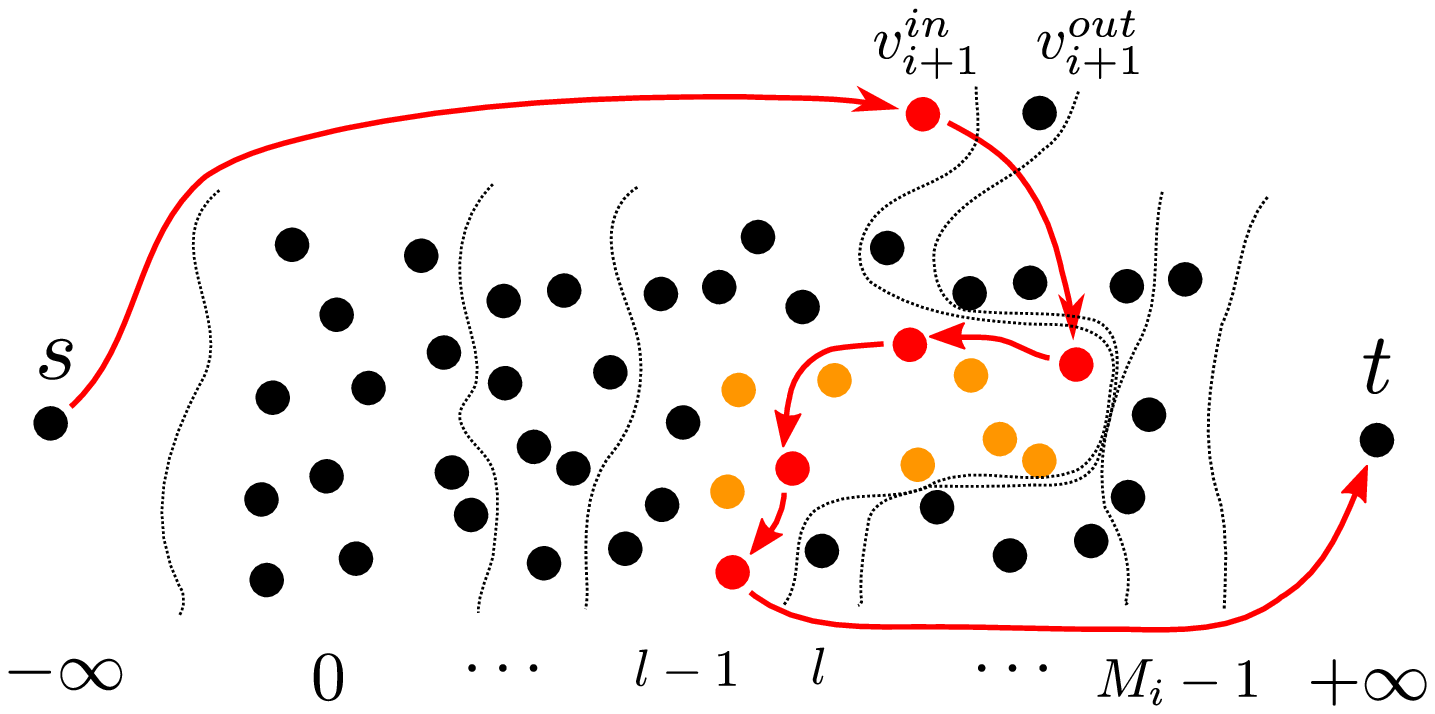}
        \caption[]%
        {{\small Merge of layer $l$}}
        \label{subfig:merge}
    \end{subfigure}
    \caption[]
    {\small Execution of our algorithm in an abstract example graph. Edges and their flow are absent for simplicity. Layers are divided by dotted vertical strokes, $M_i = \max_{v\in\flowV_i\setminus\{t\}}\ell_i(v)$. \Cref{subfig:traversal} shows a decrementing path in $\residual{\flowG_{i+1}}{t_{i+1}}$ (red) found by the layered traversal as well as all vertices visited (red and orange), $l$ is the smallest layer visited. \Cref{subfig:updates} shows the updates to the level assignment, all vertices visited by the traversal get level $l$, and $v_{i+1}^{out}$ gets level $l+1$. \Cref{subfig:merge} shows the result of merging layer $l$, all vertices of level $l$ or more decrease their level by one.}
    \label{fig:progressive}
\end{figure}

Finally, for the path updates we proceed as follows. If no decrementing path was found then we set $\pathcover_{i+1} \gets \tmp_{i+1}$ and also $path(v_{i+1}) \gets |f^*_{i+1}|$. Otherwise, we decompose the flow $f^*_{i+1}$ in the vertices of $L^{>l-1}_{i+1}$ and join it to $\tmp_{i+1} \cap L^{\le l-1}_{i+1}$. The correctness of the previous idea follows from the fact that $f^*_{i+1}(u,v) = t_{i+1}(u,v)$ when $u \in L^{\le l-1}_{i+1}$. Analogously, it is not necessary to change $path(v)$ for $v^{in} \in L^{\le l-1}_{i+1}$. To decompose $f^*_{i+1}$ in $L^{>l-1}_{i+1}$ we start decomposing each path from $t$ back to a vertex $v$ of level $\le l-1$, once we get to $v$ we join the decomposed suffix path $S'$ to the corresponding path from $\tmp_{i+1}$ ending at $v$, and set $path(u) \gets path(v)$ for $u \in S'$. Since every vertex in $L^{>l-1}_{i+1}$ can belong to up to $|f^*_{i+1}|$ paths, the running time of the entire procedure can be bounded by $O(|f^*_{i+1}|)$ per vertex in $L^{>l-1}_{i+1}$, which dominates the running time of the layered traversal.

\subsection{Running time analysis}\label{sec:running-time}
Note that the running time of step $i+1$ is bounded by $O(|N^{-}(v_{i+1})|)$ (from sparsification) plus $O(|f^*_{i+1}|) = O(k)$ per vertex whose level is $l$ or more, where $l$ is the smallest level visited by the layered traversal in  $\residual{\flowG_{i+1}}{t_{i+1}}$. The first part adds up to $O(|E|)$ for the entire algorithm, whereas for the second part we show that every vertex is charged $O(k)$ only $O(k^2)$ times in the entire algorithm, thus adding up to $O(k^3|V|)$ in total. Let $i$ be the iteration where a vertex $u$ was added, and consider the sequence $\left(d(L_{i}^{\le 0}), d(L_{i}^{\le 1}),\ldots,d(L_{i}^{\le \ell_{i}(u)})\right)$ and its evolution until $\left(d(L_{|V|}^{\le 0}), d(L_{|V|}^{\le 1}),\ldots,d(L_{|V|}^{\le \ell_{|V|}(u)})\right)$.  Note that every time $u$ is charged $O(k)$, then the smallest level visited in that step must be $l\le \ell(u)$. Moreover, by \Cref{lemma:level-path-changes}, any update that charges $u$ changes exactly one value in the sequence ($d(L^{\le l})$ is incremented by one), and possibly truncates the sequence on the right due to $u$'s level being lowered. By \invC{}, the sequence is always strictly decreasing, and since $d(L^{\le 0}) \le k$, it is updated $O(k^2)$ times.

\subsection{The improved algorithm}
We start the description of the $O(k^2|V| + |E|)$ time algorithm with a more refined analysis to that of \Cref{sec:running-time}. 

We have established that the running time of the layered traversal (\Cref{sec:layered-traversal}) and subsequent updates (\Cref{sec:level-updates}) of the iteration $i+1$, is bounded by $O(k)$ per vertex of level $l$ or more, where $l$ is the smallest visited level. However, we can obtain tighter bounds by breaking down the running time analysis as follows:
\begin{enumerate}
    \item[a)] $O(k)$ per vertex $u$ with $\ell_{i+1}(u) = l$. For the running time of the layered traversal, the update of $f^*_{i+1}$, and the update of $\ell_{i+1}$ before a possible merge.
    \item[b)] $O(1)$ per vertex of level $l$ or more. For the running time of the merge of level $l$.
    \item[c)] $O(k)$ per vertex of level $l$ or more. For the running time of the path updates that comes from decomposing $f^*_{i+1}$ into $\pathcover_{i+1}$.
\end{enumerate}

\begin{obs}
Running times a) and b) sum up to $O(k^2|V|)$.
\end{obs}
\begin{proof}
For running time a) note that every time a vertex $u$ is charged in some iteration $i+1$ $\ell_{i+1}(u) = l$. Then, by \Cref{obs:last-vertices-dont-change-level}, $d(L_{i+1}^{\le \ell_{i+1}(u)}) \ge d(L_{i}^{\le \ell_{i}(u)}) + 1$\footnote{The inequality can be strict. For example, if $\ell_{i+1}(u) < \ell_{i}(u)$.}, and since $d(L^{\le \ell(u)}) \le k$, $u$ is charged $O(k)$ times. As for running time b) we repeat the argument given in \Cref{sec:running-time}, but now every charge is only $O(1)$.
\end{proof}

To obtain our improved $O(k^2|V|+|E|)$ time algorithm we will compute the $path$ function (path ID of some path in $\pathcover_{i+1}$ containing a vertex) without maintaining $\pathcover_{i+1}$. More specifically, we show that we can maintain the $path$ function by decomposing $f^*_{i+1}$ only in edges of level $l$\footnote{At least one of the vertices has level $l$.}, which adds up to $O(k^2|V|)$ as in running time a), plus $O(1)$ per vertex of level $l$ or more, which also adds up to $O(k^2|V|)$ as in running time b). To achieve this, we will first show that the layered cuts maintained by our algorithm correspond to a structure of decreasing size antichains.

\subsubsection{Structure of antichains}\label{sec:antichain-structure}

Let us fix the iteration $i$ of the algorithm. Recall that the $l$-th layered cut is a ow-cut defined as $L_{i}^{\le l} = \bigcupdot_{j = -\infty}^{l} L_{i}^{j}$, and its complement $L_{i}^{>l} = \flowV_i \setminus L_{i}^{\le l} = \bigcupdot_{j = l+1}^{\infty} L_{i}^{j}$ (see \Cref{sec:levels-and-invariants}). By construction (of $\flowG_i$), $d(L_{i}^{\le l}) = |\{(u^{in}, u^{out}) \in \flowE_{i} \mid \ell_{i}(u^{in}) \le l \land \ell_{i}(u^{out}) > l\}|$, thus we can define $A_{i}^{l}$ as the vertices (in $G_{i}$) represented by those edges, that is $A_{i}^{l} = \{u \in V_{i} \mid \ell_{i}(u^{in}) \le l \land \ell_{i}(u^{out}) > l\}$.

\begin{obs}
$A_{i}^{l}$ is an antichain in $G_{i}$.
\end{obs}
\begin{proof}
Since $L_{i}^{\le l}$ is a ow-cut the edges exiting the cut (including those representing vertices in $A_{i}^{l}$) form an antichain (in $\flowG_i$), since a path connecting two such edges contradicts the cut being one way. This relation between ow-cuts and antichains has been studied before, see e.g.~\cite{mohring1985algorithmic,rademaker2012optimal,pijls2013another,marchal2018parallel}.
\end{proof}

\begin{figure}[t]
    \centering
    \begin{subfigure}[b]{0.48\textwidth}
        \centering
        \includegraphics[width=0.95\textwidth]{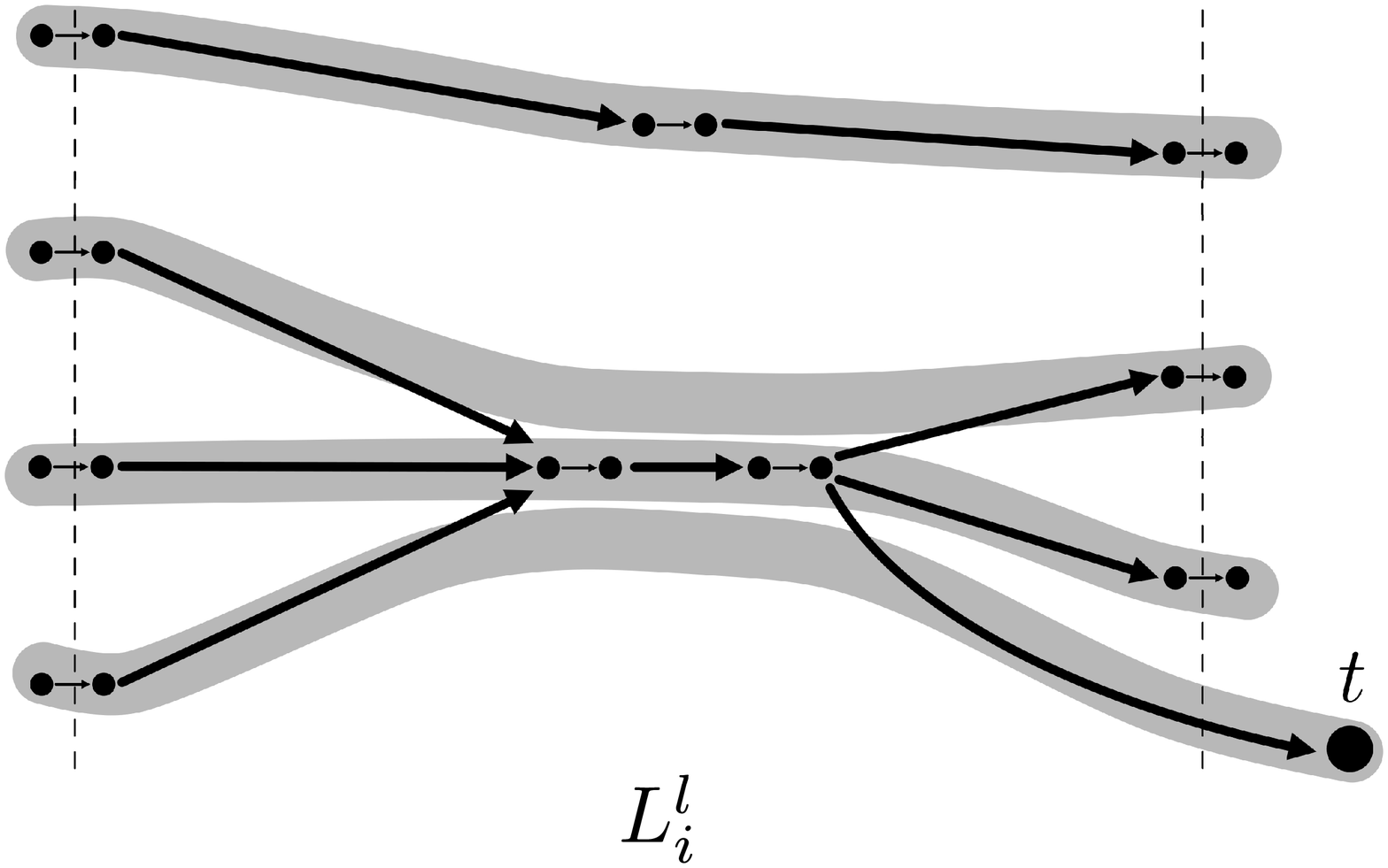}
        \caption[]%
        {{\small Flow decomposition of layer $l$}}
        \label{subfig:layerDecomposition}
    \end{subfigure}
    \begin{subfigure}[b]{0.48\textwidth}
        \centering
        \includegraphics[width=0.95\textwidth]{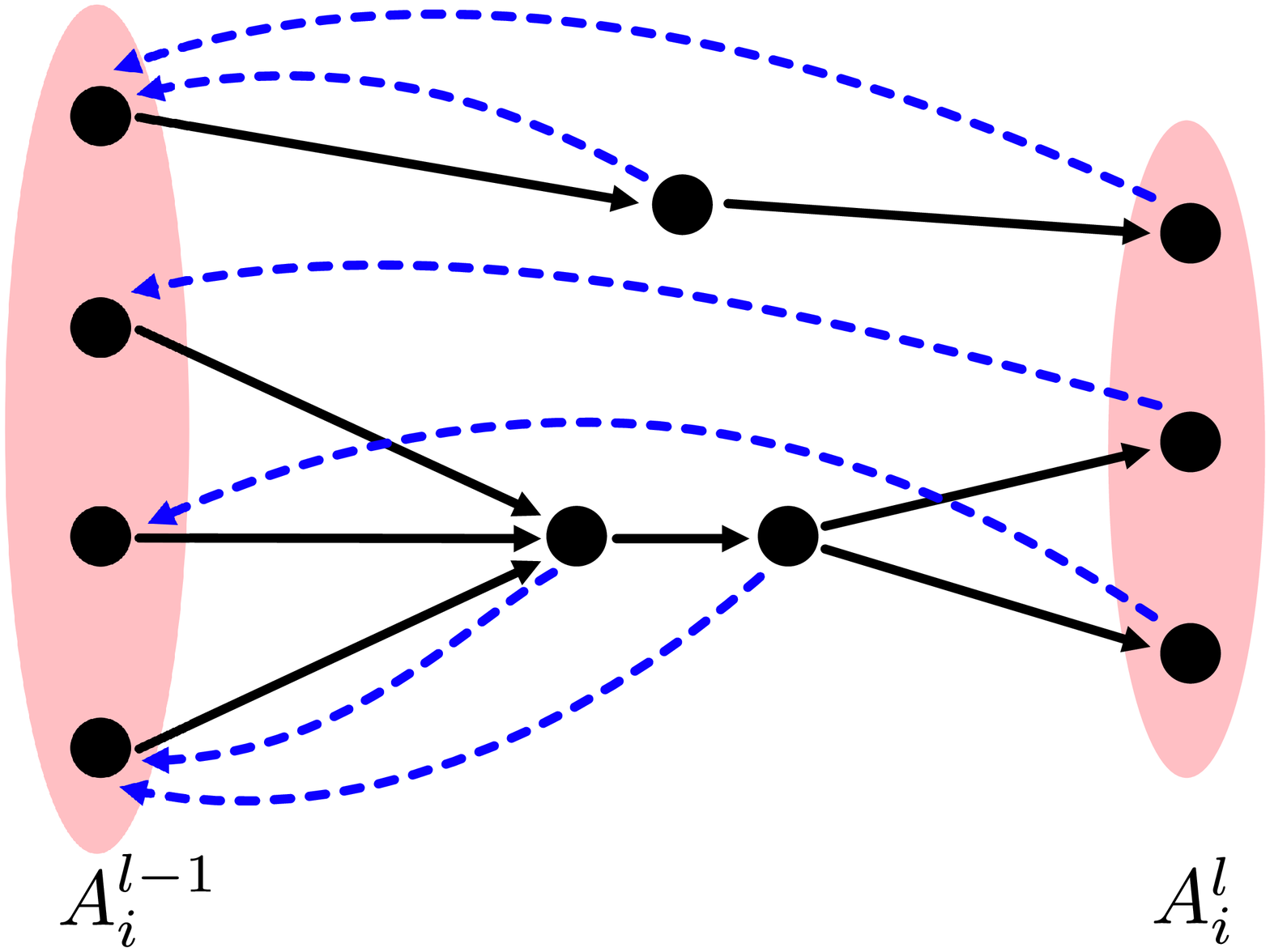}
        \caption[]%
        {{\small Layered antichains and back links}}
        \label{subfig:antichainAndLinks}
    \end{subfigure}
    \caption[]
    {\small Relation between $\pathcover_{i}$ (not maintained), the layered antichains and the back links in a layer $l$. \Cref{subfig:layerDecomposition} shows the vertices in layer $l$ (delimited within dashed vertical lines) and a flow decomposition of that layer (paths highlighted in gray). \Cref{subfig:antichainAndLinks} shows the corresponding vertices in $G$, the layered antichains $A_{i}^{l-1}, A_{i}^{l}$ (red ovals) and the back links of every vertex (blue dashed arrows) as obtained by running our maintenance procedure (see \Cref{sec:backLinksMaintenance}) when decomposing the paths from top to bottom.}
    \label{fig:layederAntichainsAndBackLinks}
\end{figure}

\invC{} implies that $|A_{i}^{l}| > |A_{i}^{l'}|$ if $0 \le l < l' \le M_i-1$. Therefore, $A_{i}^{0}, \ldots, A_{i}^{M_i-1}$ is a sequence of size-decreasing antichains implicitly maintained by our algorithm. We call these antichains, \emph{layered antichains}, and the vertices on them, \emph{antichain vertices}. By \invA{}, antichain vertices belong to exactly one path of $\pathcover_{i}$. Our improved algorithm will maintain the (only possible) path ID on antichain vertices, and for the rest of vertices $v$ it will maintain a \emph{back link} to an antichain vertex $u = backlink(v)$, such that $u$ is a predecessor of $v$ in some path of $\pathcover_{i}$. As such, we can compute the $path$ function in $O(1)$ time either directly or by first taking the corresponding back link. Note that a vertex $v$ is an antichain vertex if and only if $\ell_{i}(v^{in}) < \ell_{i}(v^{out})$, which can be checked in $O(1)$ time. See \Cref{fig:layederAntichainsAndBackLinks} for an example of these concepts.

\subsubsection{Back link maintenance}\label{sec:backLinksMaintenance} After applying the flow and level updates described in \Cref{sec:level-updates} (and only if a decrementing path was found) we decompose $f^*_{i+1}$ (only) on the edges of layer $l$, the smallest visited level. Note that a path decomposed in layer $l$ connects a vertex $u \in A_{i+1}^{l-1}$ to a vertex $w \in A_{i+1}^{l}$. After decomposing such a path, we set the back link of all vertices $v$ in the path to $backlink(v) \gets u$ and additionally we also store $newlink(v) \gets w$. Note that the running time of this decomposition is bounded by $O(k)$ per vertex of level $l$ (running time a)).

After decomposing layer $l$, we process the edges in $L_{i+1}^{>l}$ and for every edge of the form $(v^{in}, v^{out})$ we check whether the back link $u$ of $v$ is an antichain vertex. If it is not the case, we set $backlink(v) \gets newlink(u)$. Note that the running time of this procedure is bounded by $O(1)$ per vertex of level $l$ or more (running time b)). The correctness of this procedure is given by the following lemma.

\begin{lemma}\label{lemma:backlinks-decomposition}
After applying the back link maintenance procedure it holds that there is a path in $\pathcover_{i+1}$ where $backlink(v)$ is predecessor of $v$, for all $v\in V_{i+1}$.
\end{lemma}
\begin{proof}
If $v^{out} \in L_{i+1}^{\le l-1}$, then the property (inductively) holds since $\flowG_{i+1}[L_{i+1}^{\le l-1}]$ was not affected by iteration $i+1$. If $v^{in} \in L_{i+1}^{l}$ or $v^{out} \in L_{i+1}^{l}$, then the property holds since the edge $(v^{in}, v^{out})$ was decomposed while decomposing $f^*_{i+1}$ in layer $l$. If $v^{in}, v^{out} \in L_{i+1}^{> l}$, consider $u$ to be the back link of $v$ before the update and $u'$ the back link of $v$ after the update. If $u = u'$, then the path from $u$ to $v$ in $\pathcover_{i}$ only uses vertices in $L_{i+1}^{> l}$\footnote{At least one of the vertices in the corresponding edge is in $L_{i+1}^{> l}$.}, and the property (inductively) holds since $\flowG_{i+1}[L_{i+1}^{> l}]$ was not affected by iteration $i+1$. If $u \not = u'$, then $(u^{in}, u^{out})$ was decomposed while decomposing $f^*_{i+1}$ in layer $l$. This decomposition led to edge $(u'^{in}, u'^{out})$, set $newlink(u) \gets u'$ and later set $backlink(v) \gets newlink(u) = u'$ (since $u$ is no longer an antichain vertex).
\end{proof}

Finally, we update the path ID of the antichain vertices by using their back links. More specifically, for each antichain vertex $v \in A_{i+1}^{l'}$ (these can be discovered during the back link maintenance procedure) whose path ID has not been updated and such that $l' \ge l$, we iteratively take back links until reaching a vertex $u \in A_{i+1}^{l''}$ with $l'' < l$ (or whose path ID was already updated), and update the path ID of all such vertices (including $v$) to the path ID of $u$. Note that this procedure runs in $O(1)$ time per antichain vertex, and since there are $O(k^2)$ such vertices ($O(k)$ antichains of size $O(k)$ each), it runs in $O(k^2)$ per iteration, thus $O(k^2|V|)$ for the whole algorithm.

\section{Support Sparsification Algorithm}
\label{sec:edge-thinning}
We present an algorithm that transforms any path cover $\pathcover, |\pathcover| = t$ of a DAG $G = (V, E)$ into one of the same size and using less than $2|V|$ distinct edges, in $O(t^2|V|)$ time (\Cref{thm:edge-thinning}). The main approach consists of splicing paths so that edges are removed from the \emph{support} $E_{\pathcover} = \{e \in P \mid P \in \pathcover\}$. It maintains a path cover $\pathcover', |\pathcover'| = t$ of $G' = (V, E_{\pathcover'})$ (thus also a path cover of $G$). At the beginning we initialize $\pathcover'\gets \pathcover$, and we splice paths so that at the end $|E_{\pathcover'}| < 2|V|$.

\subsection{Splicing}\label{sec:splicing} We call \emph{splicing} of a path cover $\pathcover$ through a path $D$ to the process of reconnecting paths in $\pathcover$ so that (after reconnecting) at least one of the paths contains $D$ as a subpath. Splicing additionally requires that for every edge $e$ of $D$ there is at least one path in $\pathcover$ containing $e$, but also maintains the \emph{multiplicity} of the edges. Recall that the \emph{multiplicity} of an edge $e$ with respect to a set of paths $\pathcover$, $\mu_{\pathcover}(e)$, is the number of paths in $\pathcover$ using $e$, that is, $\mu(e) = |\{P \in \pathcover \mid e \in P\}|$. We show how to splice $\pathcover$ in time $O(|D|)$.

\begin{restatable}{lemma}{splicingalgorithm}
\label{result:splicing}
    Let $G = (V, E)$ be a DAG, $D$ a proper path of $G$, and $\pathcover$ a path cover of $G$ such that for every edge $e \in D$ there exists $P\in\pathcover$ containing $e$. We can obtain, in time $O(|D|)$, a path cover $\pathcover'$ of $G$ such that $|\pathcover'| = |\pathcover|$ and there exists $P\in\pathcover'$ containing $D$ as a subpath. Moreover, $\mu_{\pathcover}(e) = \mu_{\pathcover'}(e)$ for all $e \in E$.
\end{restatable}
\begin{proof}
    We process the edges of $D$ one by one, and maintain a path $P$ of the path cover that contains as subpath a prefix of $D$, at the end of the algorithm $P$ will contain the whole $D$ as a subpath, as required. We initialize $P$ to be some path of $\pathcover$ containing the first edge of $D$. Then, when processing the next edge $e$ of $D$, we first check if $e$ is the next edge of $P$, if so we continue to the next edge of $D$. Otherwise, let $P'$ be a path of the path cover containing $e$, then we connect the prefix of $P'$ until $e$ (excluding) with the suffix of $P$ from the edge previous to $e$ in $D$ (excluding), and we also connect the prefix of $P$ until the edge previous to $e$ in $D$ (including) with the suffix of $P'$ from $e$ (including). Note that each of these connections can be made by manipulating pointers in $O(1)$ time, also note that the new set of paths forms a path cover, and the edges of $G$ preserve their multiplicity, as edges in the path cover are never created or removed, only change path. 
\end{proof}

% Because of the last property of $\pathcover'$, the flow induced by $\pathcover$ is the same as the flow induced by $\pathcover'$. As such, if $\pathcover$ is a flow decomposition of a flow $f$, then $\pathcover'$ is also a flow decomposition of $f$.

To decide how to splice paths in our sparsification algorithm, we color the vertices of $v \in V$ based on their degree, that is, if $deg_{G'}(v) \le 2$ we color $v$ \blue{}, and \red{} otherwise. We also color the edges $(u, v) \in E_{\pathcover'}$ according to the color of their endpoints, that is, if both $u$ and $v$ are \blue{}, we color $(u, v)$ \blue{}, likewise if both $u$ and $v$ are \red{}, we color $(u, v)$ \red{}, otherwise we color $(u, v)$ \purple{}. We traverse the underlying undirected graph of $G'$ in search of a \red{} cycle  (cycle of \red{} edges) $C$ and splice paths along $C$ so that at least one \red{} edge is removed from $G'$. We repeat this until no \red{} cycles remain. Therefore, at the end we have that \red{} vertices and edges form a forest, \blue{} vertices and edges form a collection of vertex-disjoint paths and cycles, and \purple{} edges connect \red{} vertices with the extreme vertices of \blue{} paths. As such, if the number of \blue{} and \red{} vertices is $n_b$ and $n_r$, respectively, and the number of \blue{} paths is $p$, there are $n_b - p$ \blue{} edges, less than $n_r$ \red{} edges, and at most $2p$ \purple{} edges. Therefore, $|E_{\pathcover'}| < n_b -p + n_r + 2p \le 2|V|$, as desired. The following remark shows that the factor $2$ from the bound is asymptotically tight.

\begin{figure}[t]
    \centering
    \includegraphics[scale=0.4]{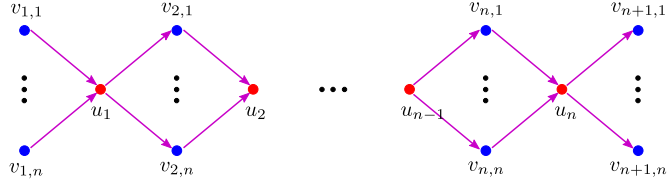}
    \caption{\small A DAG $G$ showing that the factor $2$ from the bound of \Cref{thm:edge-thinning} is asymptotically tight. The figure shows the example graph, as well as the result of applying \Cref{thm:edge-thinning} on an MPC of it. The algorithm colors vertices $v_{i,j}$ \blue{}, vertices $u_i$ \red{}, and edges \purple{}, thus it does not find any \red{} cycle.}
    \label{fig:2-tight}
\end{figure}

\begin{remark}\label{remark:2-tight}
Consider the DAG $G = (V,E)$ from \Cref{fig:2-tight}, with $|V| = n+ n(n+1) = n(n+2)$, $|E| = 2n^2$ and $\width(G) = n$. Note that any path cover $\pathcover$ of size $n$ (an MPC) must use every edge of the graph, then $|E_{\pathcover}|/|V| = |E_{\pathcover'}|/|V| = 2 - 4/(n+2)$.
\end{remark}

\begin{figure}[t]
    \centering
    \begin{subfigure}[b]{0.32\textwidth}
        \centering
        \includegraphics[width=0.95\textwidth]{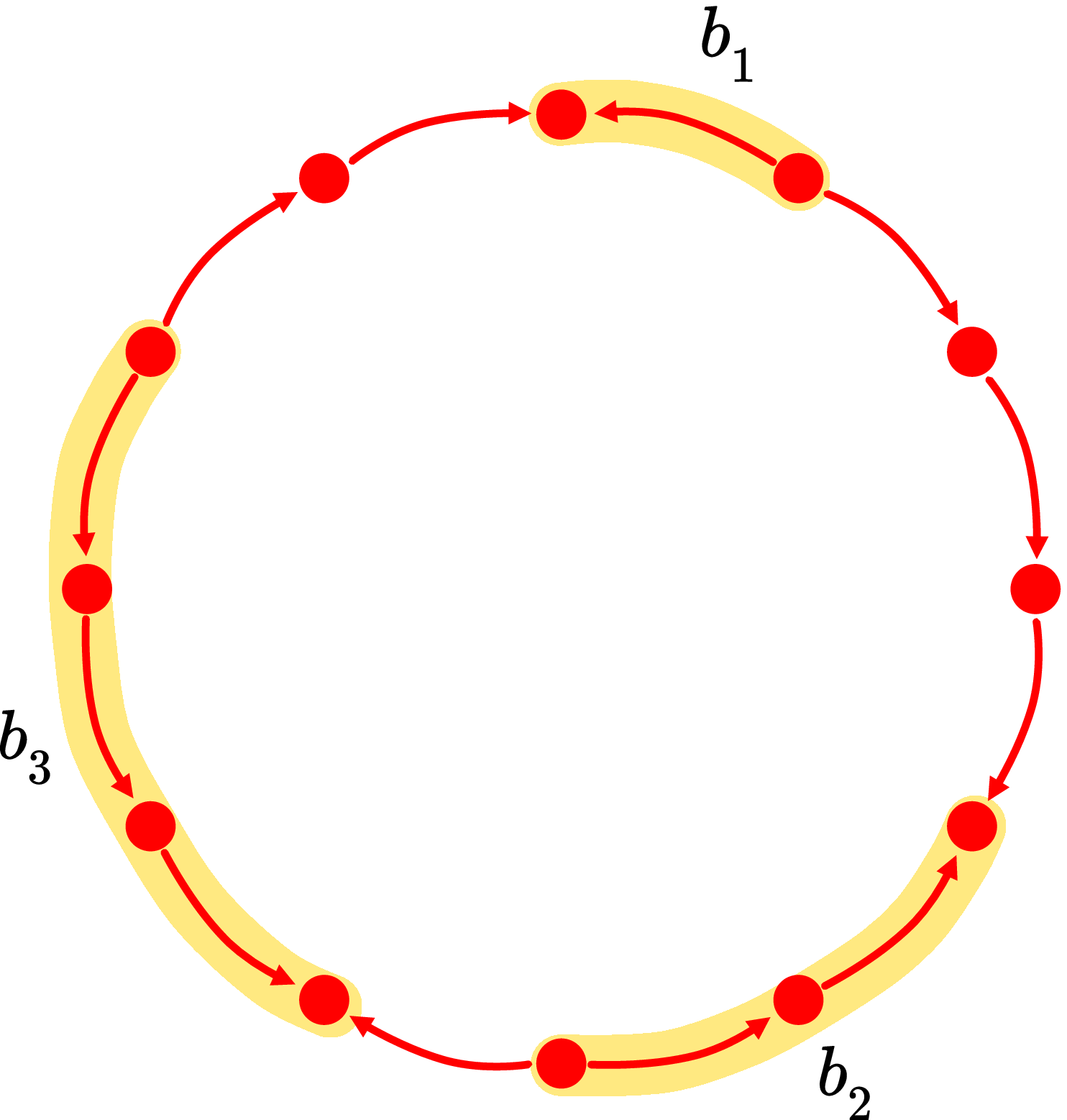}
        \caption[]%
        {{\small Backward segments}}
        \label{subfig:splicing1}
    \end{subfigure}
    \begin{subfigure}[b]{0.32\textwidth}
        \centering
        \includegraphics[width=0.95\textwidth]{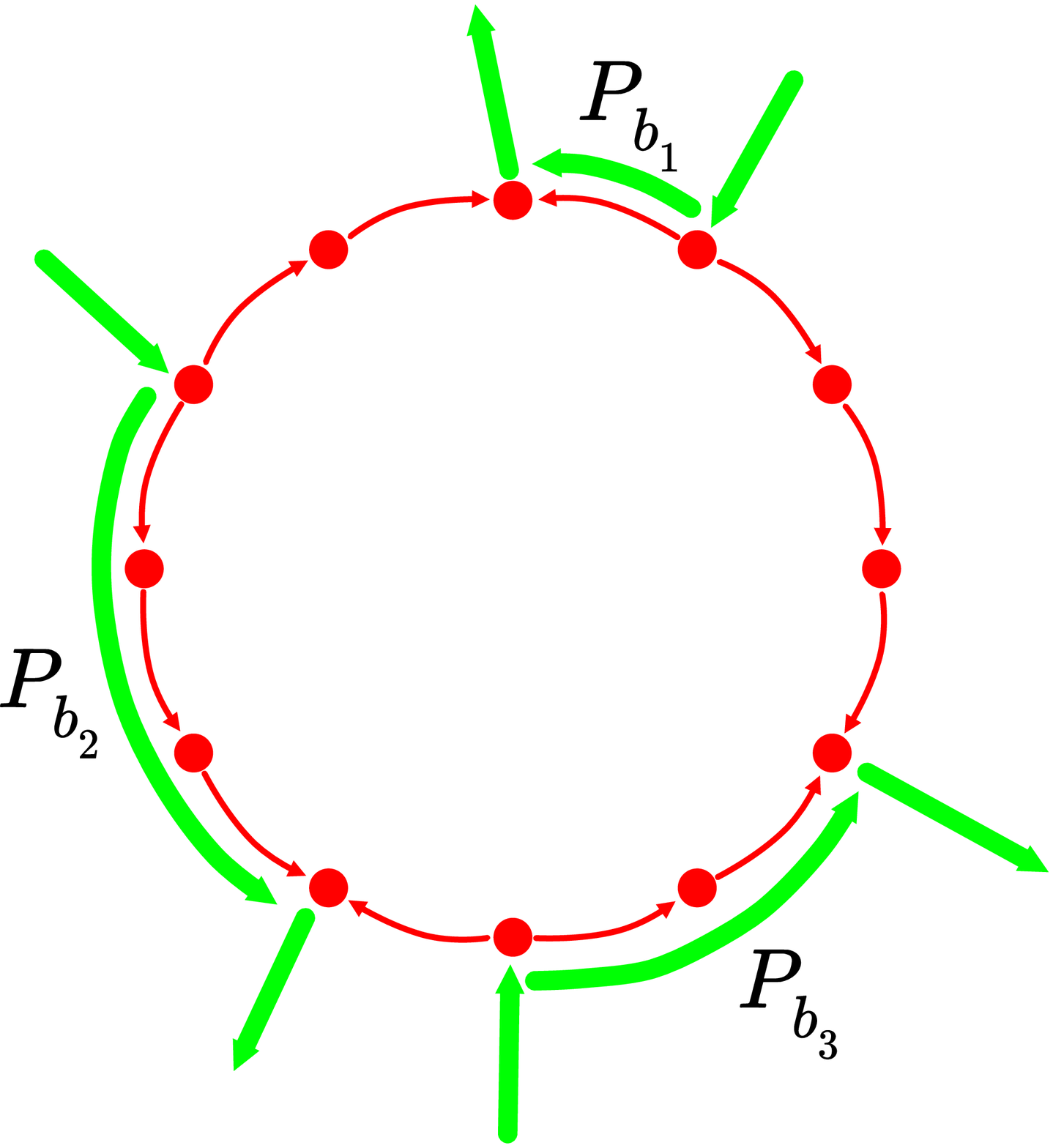}
        \caption[]%
        {{\small Spliced paths}}
        \label{subfig:splicing2}
    \end{subfigure}
    \begin{subfigure}[b]{0.32\textwidth}
        \centering
        \includegraphics[width=0.95\textwidth]{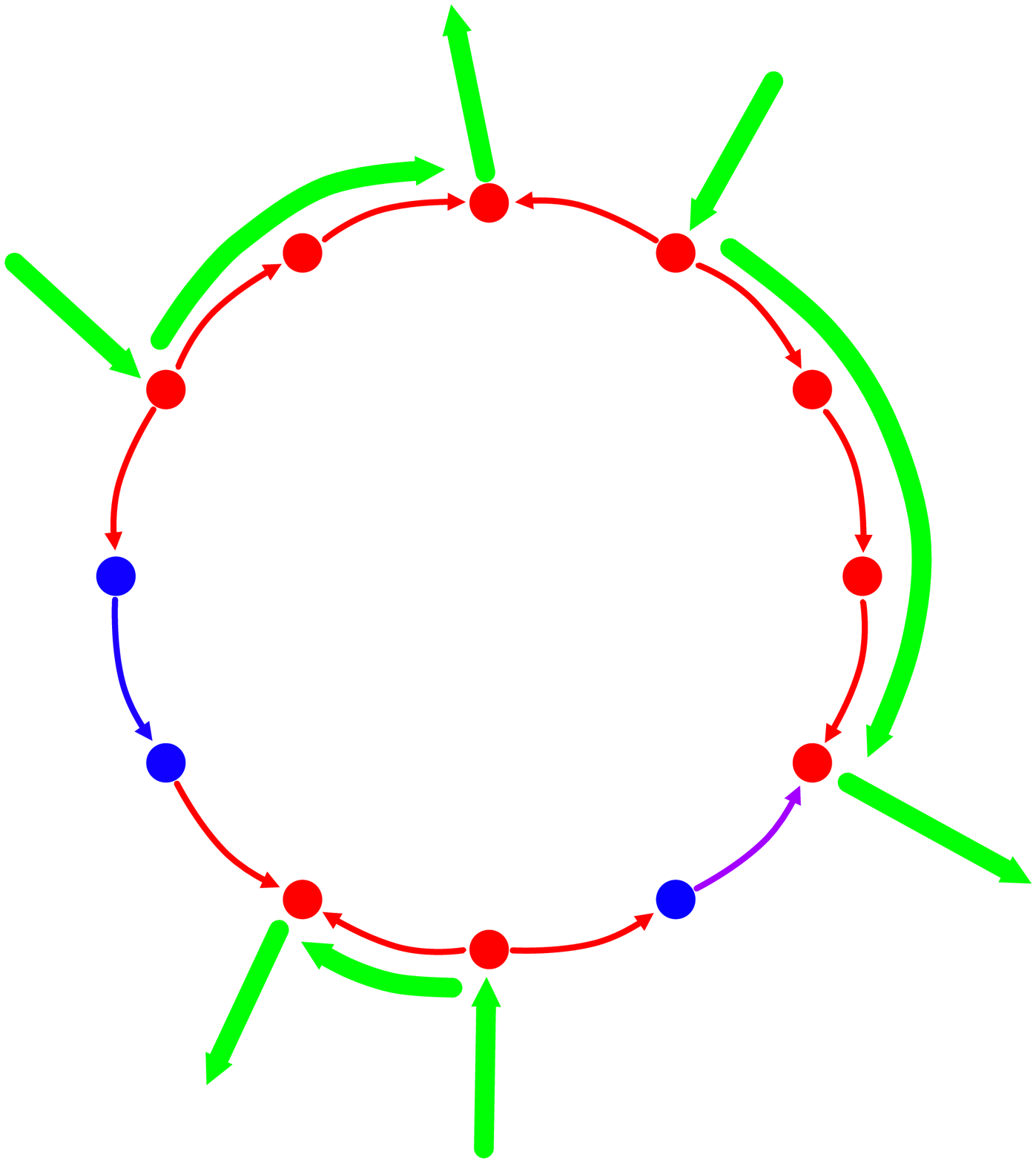}
        \caption[]%
        {{\small New paths}}
        \label{subfig:splincing3}
    \end{subfigure}
    \caption[]
    {\small Splicing in a red cycle. \Cref{subfig:splicing1} shows a red cycle with three backward segments highlighted in yellow (clockwise orientation of the cycle). \Cref{subfig:splicing2} shows the paths containing the backward segments (in green) after splicing through them. \Cref{subfig:splincing3} shows the new pathsPath reconnecting (in green) after removing the backward segments and reconnecting the suffixes and prefixes through the corresponding forward segments, it also shows that some vertices and edges become blue and purple, removing the red cycle.}
    \label{fig:splicing}
\end{figure}

When processing a \red{} cycle $C = v_1, \ldots, v_l, v_{l+1} = v_1$, we partition the corresponding edges of $G'$ in either \emph{forward} $F = \{(v_{i}, v_{i+1}) \in E_{\pathcover'}\}$ or \emph{backward} $B = \{(v_{i+1}, v_{i}) \in E_{\pathcover'}\}$ edges. We splice either along forward or backward edges depending on the comparison between $\sum_{e \in F} \mu(e)$ and $\sum_{e \in B} \mu(e)$. If $\sum_{e \in F} \mu(e) \ge \sum_{e \in B} \mu(e)$, we only splice along backward edges, otherwise only along forward edges. Here we only describe how to splice along backward edges, splicing along forward edges is analogous.
% If $\mu_b$ is the minimum multiplicity among backward edges ($\mu_b = \min_{e \in B} \mu(e)$), then we repeat the splicing procedure $\mu_b$ times, each of which will decrease the multiplicity of backward edges by one and increase the multiplicity of forward edges by one. As such, after $\mu_b$ applications of the splicing procedure, all backward edges with\ multiplicity $\mu_b$ will be removed from $G'$. 
The splicing procedure considers the backward \emph{segments} of the cycle, namely, maximal subpaths of consecutive backward edges in $C$. For each backward segment $b$, it generates a path $P_b \in \pathcover'$ that traverses $b$ entirely, by splicing paths along $b$. For this we apply the splicing procedure of \Cref{result:splicing} on every backward segment, which runs in total time $O(|B|) = O(|C|)$. After that, for every $P_b$ we remove $b$ and reconnect the parts of $P_b$ entering and exiting $b$ to their corresponding adjacent forward segments. See \Cref{fig:splicing} for an example. Note that vertices of $b$ are still covered by some path after splicing since they are \red{}, and the splicing procedure preserves the multiplicity of edges. Also note that the net effect is that the number of paths remains unchanged but the multiplicity of forward edges has increased by one and the multiplicity of backward edges has decreased by one, thus the condition $\sum_{e \in F} \mu(e) \ge \sum_{e \in B} \mu(e)$ will be valid again after the procedure. As such, we repeat the splicing procedure until some backward edge has multiplicity $0$, removing $C$ in this way.

To analyze the running time of all splicing procedures during the algorithm, we consider the function $\Phi(G') = \sum_{e \in E_{\pathcover'}} \mu(e)^2$. We study the change of $\Phi(G')$ of applying the splicing procedure, $\Delta \Phi$. Since the only changes on multiplicity occur on forward and backward edges we have that 
\begin{align*}
\Delta \Phi &= \sum_{e\in F} \left((\mu(e)+1)^2-\mu(e)^2\right) + \sum_{e\in B} \left((\mu(e)-1)^2-\mu(e)^2\right)\\
&= |F| + |B| + 2 \left(\sum_{e \in F} \mu(e) - \sum_{e \in B} \mu(e)\right) \ge |C|.
\end{align*}

As such, each splicing procedure takes $O(|C|)$ time, and increases $\Phi(G')$ by at least $|C|$. Since at the end $\Phi(G') \le  t^2|E_{\pathcover'}|  \le t^22|V|$, the running time of all splicing procedures amounts to $O(t^2|V|)$.

Finally, we describe how to traverse the underlying undirected graph of $G'$ while detecting \red{} cycles in linear time, which is $O(t|V|)$. We perform a \emph{modified} DFS traversal of the graph. We additionally mark the edges as \emph{\processed{}} either when the edge is removed (gets multiplicity $0$), or when the traversal \emph{pops} this edge from the DFS stack\footnote{When an edge is marked as \processed{} we move it at the end of the adjacency list of the corresponding vertex. Therefore, the first edge in the adjacency list of a vertex is always not marked as \processed{}, unless all of them are.}. Since our graph is undirected, all edges are between a vertex and some ancestor in the DFS tree (no crossing edges), thus cycles can be detected by checking if the vertex being visited already is in the DFS stack (and it is not the top of the stack)\footnote{We can maintain an array \emph{in-stack} indicating whether a vertex is in the DFS stack.}. When a \red{} cycle is detected, then we pop from the DFS stack all vertices of the cycle, but without marking as \processed{} the corresponding edges. The cost of these pops plus the additional cost of traversing the edges of the cycle again in a future traversal is linear in the length of the cycle, thus these are charged to the corresponding splicing procedures of this cycle, and the cost of the traversal remains proportional to the size of the graph. 

% The following lemma establishes the correctness of this approach.

% \begin{lemma}
% If an edge $e \in E_{\pathcover}$ is marked as \processed{} by the traversal, then there are no \red{} cycles in $G'$ using this edge.
% \end{lemma}
% \begin{proof}
%     If $e$ was removed from the graph, the claim trivially follows. Otherwise, $e$ was popped from the DFS stack by the traversal. Suppose by contradiction that there is a \red{} cycle $C$ containing $e$, an consider the first edge $e' = (x, y)$ (orientation given by the DFS traversal) of $C$ that was popped from the stack. Now, consider the segment $P$ (in $C$) containing $e'$ of vertices in the DFS stack when $e'$ is pushed, and let $u$ be the first vertex in this segment (whereas $y$ is the last). Because $C$ exists and $e'$ is the first edge popped, at this point there is a \red{} path between $y$ and $u$ without using $e'$, which should have been traversed by the algorithm, a contradiction.
% \end{proof}

\bibliographystyle{siamplain}
\bibliography{main}

\end{document}